\documentclass[10pt]{article}
\usepackage{amsmath, amsthm, amssymb, pdfsync, psfrag, verbatim,color} 
\usepackage{hyperref}
\usepackage {mathrsfs, graphicx, newcent,wrapfig,multirow}
\usepackage[round]{natbib}
\usepackage{breakcites}
\usepackage[affil-it]{authblk}
\usepackage{epsfig}

\newcommand{\be}{\begin{equation}}
	\newcommand{\ee}{\end{equation}}
\newcommand{\bea}{\begin{eqnarray}}
	\newcommand{\eea}{\end{eqnarray}}
\newcommand{\beas}{\begin{eqnarray*}}
	\newcommand{\eeas}{\end{eqnarray*}}
\newcommand{\eq}[1]{\begin{equation}\begin{aligned}#1\end{aligned}\end{equation}}
\newcommand{\eqn}[1]{\begin{align*}#1\end{align*}}

\newcommand{\bbE}{\mathbb E}
\newcommand{\bbF}{\mathbb F}

\newcommand{\bbI}{\mathbb I}

\newcommand{\bbN}{\mathbb N}

\newcommand{\bbP}{\mathbb P}

\newcommand{\bbR}{\mathbb R}

\newcommand{\bbZ}{\mathbb Z}

\newcommand{\scC}{\mathcal C}

\newcommand{\scF}{\mathcal F}

\newcommand{\scM}{\mathcal M}

\newcommand{\scR}{\mathcal R}

\newcommand{\scU}{\mathcal U}

\newcommand{\veps}{\varepsilon}

\newcommand{\norm}[1]{\ensuremath{\left\| #1 \right\|}}
\newcommand{\abs}[1]{\ensuremath{\left| #1 \right|}}

\newcommand{\indicator}[1]{\ensuremath{\mathbb{I}_{#1}}}

\newcommand*\widebar[1]{%
	\hbox{%
		\vbox{%
			\hrule height 0.5pt 
			\kern0.5ex
			\hbox{%
				\kern-0.1em
				\ensuremath{#1}%
				\kern-0.1em
			}%
		}%
	}%
} 

\newcommand{\crl}[1]{\ensuremath{ \left\{ #1 \right\} }}
\newcommand{\edg}[1]{\ensuremath{ \left[ #1 \right] }}
\newcommand{\brak}[1]{\ensuremath{\left( #1 \right)}}

\newtheorem{assumption}{Assumption}
\newtheorem{theorem}{Theorem}[section]
\newtheorem*{algo}{RWPI Protocol}
\newtheorem*{prob}{Problem}

\newtheorem{proposition}[theorem]{Proposition}

\newtheorem{remark}[theorem]{Remark}
\newtheorem{example}[theorem]{Example}
\newtheorem{examples}[theorem]{Examples}
\newtheorem{foo}[theorem]{Remarks}

\newenvironment{Remark}{\begin{remark}\rm}{\end{remark}}

\usepackage{graphicx} 
\usepackage{float} 
\usepackage{subfigure} 
\allowdisplaybreaks[4]
\usepackage[title]{appendix}
\title{Data-driven Multiperiod Robust Mean-Variance Optimization}
\author{
Xin Hai \&	Gregoire Loeper \& Kihun Nam\\ 
	Monash University\\
	Clayton, VIC 3800, Australia
}
\date{\today}

\begin{document}
\maketitle
\begin{abstract}
We study robust mean-variance optimization in multiperiod portfolio selection by allowing the true probability measure to be inside a Wasserstein ball centered at the empirical probability measure. Given the confidence level, the radius of the Wasserstein ball is determined by the empirical data. The numerical simulations of the US stock market provide a promising result compared to other popular strategies.\\[2mm]
{\bf Keywords:} Mean-Variance, Robust Portfolio Selection, Wasserstein Distance, Modern Portfolio Theory\\[2mm]
\end{abstract}
\setcounter{equation}{0}
\section{Introduction}
\label{section1}
In this article, we study robust mean-variance optimization in multiperiod portfolio selection. In particular, we allow the true probability measure to be inside a Wasserstein ball specified by the empirical data and the given confidence level. We transform our optimization problem into a non-robust minimization problem with a penalty, which provides a tractable model for robust mean-variance optimization. This extends the single-period model of \cite{blanchet2021distributionally} to a multiperiod model. Then, we apply our framework to the US stock market on five different 10-year intervals between 2002 and 2019, which provides Sharpe ratios competitive with the equal-weighted portfolio, the classical Markowitz portfolio, and various robust portfolio strategies.

Since the seminal result of \cite{markowitz1952portfolio} on the static mean-variance portfolio optimization, the portfolio optimization theory has received much attention from academics and the industry. Mean-variance optimization has become popular because it has a clear and informative target function, which contains explicit profit and risk terms and describes the trade-off between them. Unfortunately, the model is extremely sensitive to the mean and covariance of the underlying stocks, and both the empirical mean and covariance can deviate significantly from the true ones.

To solve this issue, robust formulations are developed in the literature, allowing the possibility that the true probability measure may differ from the empirical one. For example, \cite{lobo2000worst} provided a worst-case robust analysis where the mean and covariance of returns belong to prescribed sets. \cite{pflug2007ambiguity} studied the robust Markowitz problem when the true probability distribution is not perfectly known. They used a confidence set for the probability distribution to find the robustness. However, both papers did not provide guidance for an appropriate selection of the uncertainty set's size. In addition to the discrete-time portfolio selection problem, \cite{pham2022portfolio} proposed a separation principle for solving the robust mean-variance problem in continuous time, when the expected return and correlation matrix are ambiguous.

The conventional approach to obtain a robust solution involves 1) specification (shape and size) of an ambiguity set, which is the collection of candidate probability measures for the true probability measure, and 2) analysis of the worst-case scenario.

Typical choices for the ambiguity set are a set of measures equivalent to the empirical measure\footnote{for example, the drift of stock prices can be allowed in a certain set} \citep[e.g.,][]{hansen2001robust}, an inverse image of measurable function \footnote{for example, the mean and variance can be allowed to be in certain specified intervals} \citep[e.g.,][]{lobo2000worst,tutuncu2004robust}, and a topological neighborhood of the empirical probability measure\footnote{for example, the true probability measure is allowed to be in a Wasserstein ellipsoid containing the empirical probability measure} (e.g., \citealp[Chapter 12]{fabozzi2007robust}; \citealp{hota2019data,fournier2015rate}). 
Each class of the ambiguity set has its strength and weakness. 

The ambiguity set's size is important as well. If one allows too much uncertainty, the solution would be too conservative. Conversely, if the size of the ambiguity set is too small, then the robustness is lost. Therefore, choosing the appropriate size of the ambiguity set is important. Conventionally, the one performing the optimization chooses the ambiguity set's size based on separate statistical estimations. Recently, data-driven approaches, which combine the estimation and the optimization steps, were proposed in the literature. A short list includes \cite{delage2010distributionally}, \cite{esfahani2018data}, and \cite{blanchet2019data}.

In particular, \cite{blanchet2019robust} studied the case where their ambiguity set is defined as a Wasserstein ball of probability measures. They proposed a systemic way of choosing the suitable radius of a Wasserstein ball for a given confidence level of true probability lying on the ball. Robust Wasserstein Profile Inference (RWPI) was used in \cite{blanchet2021distributionally} to obtain a data-driven robust protocol for single-period mean-variance optimization. They transformed the robust mean-variance optimization into a non-robust mean-variance optimization with a penalty function, which is computationally tractable. In addition, they empirically compared their strategies with some popular methods and achieved better average returns against most models.

Two interesting questions remain. Can we use the method for multiperiod models, and how does the strategy perform compared to other strategies? We study these questions in this article.

A multiperiod robust mean-variance model is inherently different from the single-period one. First, the mean-variance optimization is not time-consistent. In other words, optimal control in a sub-interval is not necessarily the optimal control in the whole interval. As a result, we cannot iterate the single-period model to obtain the optimal control for the multiperiod model. Another difference is the information structure. Our optimal control (investment) has to be adapted to the filtration: our optimal control at a given time will depend on the historical data up to that time.

To circumvent the problem, we represent the optimal control as a function of historical stock prices and approximate it using the Taylor expansion. Then, the multiperiod mean-variance optimization problem transforms into a high-dimensional single-period problem. While the single-period problem we obtain is not the same as \cite{blanchet2021distributionally} due to the adaptedness condition, one can check that the result in \cite{blanchet2019robust} can still be applied. As a result, we prove that the robust mean-variance optimization problem is equivalent to a non-robust optimization with a penalty term. 

We perform numerical simulations on the top 100 US stocks' adjusted price data with and without transaction costs. We selected five ten-year periods 2002-2012, 2004-2014, 2006-2016, 2008-2018, and 2009-2019. We use the first two years to estimate the parameters and the other eight years to test the model. We compared the performance (the mean, the variance, and the Sharpe ratio) of the 2- and 3- period models with the equal-weighted portfolio, the classical Markowitz model, the single-period model in \cite{blanchet2021distributionally}, and the robust models used in \cite{xing2014robust,dai2019sparse, xidonas2017robust}; \citet[Chapter 12]{fabozzi2007robust}; and \cite{won2020robust}.  To compare the $n$-period models with the single-period model, we split the portfolio into $n$ sub-portfolios starting at day one, day two,..., and day $n$ and track the daily wealth of the portfolio. 

The numerical results are impressive. The 2-period model outperforms most non-robust and robust strategies other than the 3-period model with and without the transaction costs. The only exception is the 2008-2018 simulation: the 2-period model showed the moderately high Sharpe ratio and outperformed the equal-weighted, NC, and SP strategies. For other simulations, the 2-period model showed the highest Sharpe ratio and the lowest risk. When we analyzed similarly for the 3-period model using the top 50 stocks, the 3-period model outperformed the single and 2-period models except for the 2004-2014 and 2009-2019 simulations.

The rest of the paper is organized as follows. In Section \ref{section2}, we formulate the robust mean-variance optimization model. Section \ref{section3} concerns the approximation for the multi-period investment strategy. Section \ref{section4} presents the feasible region to the optimal problem. We figure out a new method to consider transaction costs and rebalancing for our multi-period model in Section \ref{section5}. In Section \ref{section6} we provide the numerical simulations and compare our 2-period model with some conventional methods. Section \ref{section7} shows the effect of the number of periods. Finally, Section \ref{section8} concludes the paper. 

\section{Problem Formulation}
\label{section2}
	
Consider a market with $n$ risky assets and one risk-free asset without an interest rate. Let $(\Omega,\scF,\bbF, P^*)$ be a filtered probability space with $n$ discrete stochastic processes $\{S^i_t: t=1,...,T\}, i=1,2,..,n$ and $\bbF$ is the filtration generated by $\{S^i:i=1,2,...,n\}$. Here, $S^i_t$ represents the $i$th risky asset price at time $t$. Let us denote $R_t^i$ the return of risky asset $i$ on the time interval $(t-1, t]$ for $t\in[1,T]$, that is $R_t^i=(S_{t}^i-S_{t-1}^i)/S_{t-1}^i $. Then, we can define a $nT$-dimensional return vector $\boldsymbol{R}=(R_1^1,...,R_1^n,R_2^1,...,R_2^n,...,R_{T}^1,$ $...,R_{T}^n)$ that encodes the path of asset prices. Throughout this paper, any vector is understood to be a column vector and the transpose of $\boldsymbol{R}$ is denoted by $\boldsymbol{R}^{\intercal}$. 
 
Throughout the article, we use the notation ${\bf 1}$ for a vector of appropriate dimension where each element is $1$. On the other hand, $\bbI_{E}$ denotes the indicator function; that is, if $E$ is satisfied, then its value is $1$, and otherwise, zero. The norm $\norm{\cdot}_p$ denotes the $l_p$ norm in the Euclidean space.
	
	We let $\pi_t^i$ be the amount of money invested in the $i$th risky asset on the time interval $(t-1, t]$ and 
	$\boldsymbol{\pi}=\left(\pi_1^1,...,\pi_1^n,\pi_2^1,...,\pi_2^n,...,\pi_{T}^1,...,\pi_{T}^n\right).$
	Since the filtration is generated by $\crl{S^i:i=1,2,...,n}$ and $\pi$ should be predictable, there exist measurable functions $f^i:\crl{1,2,..,T}\times\brak{\bbR^{n}}^T\to\bbR, i=1,2,...,n$ such that
	\begin{itemize}
		\item $f^i(T-t,r_1, r_2,...,r_T)$ does not depend on the last variables $r_{t},r_{t+1},...,r_{T}$ and
		\item For each $i=1,2,...,n$, \begin{align*}
			\pi^i_t=f^i(T-t,R_1,R_2,...,R_{T}).
		\end{align*}
	\end{itemize}
In this article, we will use the second-order Wasserstein metric on the space of probability measures
	\begin{align*}
		W_2(\mu,\nu):=\brak{\inf_{(X,Y)\in\scC(\mu,\nu)}\bbE|X-Y|^2}^{1/2},
	\end{align*}
	where
	\begin{align*}
		\scC(\mu,\nu):=\crl{(X,Y): \bbP\circ X^{-1} =\mu, \bbP\circ Y^{-1}=\nu}.
	\end{align*}
	Now let us state our problem. For a given $\delta>0$, we define an ambiguity set by $\delta$-neighborhood of $Q$, 
	\eqn{\scU_\delta(Q):=\{P:W_2(P,Q)\leq\delta\},}
	where $Q$ is the empirical probability we obtained from historical data,
	\eqn{Q(\cdot)=\frac{1}{N}\sum_{i=1}^{N}\delta_{R^{(i)}}(\cdot),
	}
	where $(R^{(i)})_{i=1, 2, ..., N}$ are realizations of $\boldsymbol{R}$ and $\delta_{R^{(i)}}(\cdot)$ is the indicator function.
	
	The feasible region of portfolio investment is given by
	\begin{align}\label{feasible_R}
		\scF^R_{(\delta,\bar{\alpha})}:=
		\left\{\boldsymbol{\pi}:\begin{array}{l}
			\min_{P\in\scU_\delta(Q)}
			E_P\left[\boldsymbol{\pi}^\intercal\boldsymbol{R}\right]\geq\bar{\alpha}, \sum_{i=1}^n\pi_t^i=1\\
			\pi^i \text{ is predictable for each } i\in\{1,2,..., n\}\\
			\text{For given $t$, }\pi^i_t\text{ are analytic functions of $\boldsymbol{R}$ for each } i\in\{1,2,..., n\}
		\end{array}
		\right\}
	\end{align}
	Here, $\bar \alpha$ is the worst acceptable expected return. 
		\begin{Remark}
		It is reasonable to restrict our investment strategy function $f^i(x,\cdot)$ to be integrable with respect to some probability measure $\lambda$ since the investment in practice has constraints in their size. Since the set of polynomial is dense in $L^1(\mathbb{R}^{nT}, \lambda)$, any integrable investment strategies satisfying $\min_{P\in\scU_\delta(Q)}
		E_P\left[\boldsymbol{\pi}^\intercal\boldsymbol{R}\right]\geq\bar{\alpha}$ can be approximated by the strategies in $\scF^R_{(\delta,\bar\alpha)}$. On the other hand, restricting our feasible region to polynomials enables us to apply the Taylor expansion technique to facilitate the computation.
	\end{Remark}
	
	For a given confidence level $1-\delta_0$, the investor needs to solve the following optimization problem,
	\begin{equation}\label{eq1_R}
		\begin{aligned}
			\inf_{\boldsymbol{\pi}\in\scF^R_{(\delta,\Bar{\alpha})}}\max_{\lambda\geq \bar\alpha}\brak{\sup_{P\in \scU_\delta(Q),E_P[\boldsymbol{\pi}^{\intercal}\boldsymbol{R}]=\lambda}E_P\edg{\brak{\boldsymbol{\pi}^\intercal\boldsymbol{R}}^2-\lambda^2}}
		\end{aligned}
	\end{equation}
where $\delta$ and $\bar\alpha$ are chosen so that 
\[\Pi^R_{P^*}\subset\brak{\cup_{P\in\scU_\veps(Q)}\Pi^R_P}\cap \scF^R_{(\delta,\bar\alpha)}\]
with probability more than $1-\delta_0$. Here, $\Pi^R_\mu$ is the set of all solution of \eqref{eq1_R} when $Q$ and $(\delta,\bar\alpha)$ are replaced by $\mu$ and $(0,\lambda)$, respectively. 

In order to simplify the predictability condition on $\scF^R_{(\delta,\bar\alpha)}$ to a computationally tractable model, we approximate the investment strategy using the Taylor series to approximate
$\boldsymbol{\pi}^\intercal\boldsymbol{R}\approx A^\intercal M$ and 
$\sum_{i=1}^n\pi_t^i=\sum_{i=1}^n(A^i_t)^\intercal N_t^i=1$.
The $A$ represents the coefficients of the Taylor series approximation of $\boldsymbol{\pi}^\intercal\boldsymbol{R}$ with respect to $\boldsymbol{R}$ and $M=(R_j^e,R_d^cR_j^e,R_b^aR_d^cR_j^e,\cdots)$, where $a, c, e = 1,...,n$ and $b, d, j= ts-s+1,...,ts$. Here, $R_b^a$ stands for $R_{ts-s+1}^1,...,R_{ts}^n$ and it is arranged in the order of `$abcdej$', which is the same in the terms $R_d^cR_j^e$ and $R_b^aR_d^cR_j^e$: see Section \ref{section3} and \ref{section4} for detail. Then, we set 
	\begin{align}\label{feasible}
	\scF_{(\delta,\bar{\alpha})}:=
	\left\{A: \min_{P\in\scU_\delta(Q)}
		E_P\left[A^\intercal M\right]\geq\bar{\alpha}, \sum_{i=1}^n(A^i_t)^\intercal N_t^i=1
	\right\}.
\end{align}
Under such approximation, our problem \eqref{eq1} transforms to the following problem:\\
\begin{prob}
For a given confidence level $1-\delta_0$, the investor needs to solve the following optimization problem,
\begin{equation}\label{eq1}
	\begin{aligned}
		\inf_{A\in\scF_{(\delta,\Bar{\alpha})}}\max_{\lambda\geq \bar\alpha}\brak{\sup_{P\in \scU_\delta(Q),E_P[A^{\intercal}M]=\lambda}E_P\edg{\brak{A^\intercal M}^2-\lambda^2}},
	\end{aligned}
\end{equation}
where $\delta$ and $\bar\alpha$ are chosen so that 
\[\Pi_{P^*}\subset\brak{\cup_{P\in\scU_\veps(Q)}\Pi_P}\cap \scF_{(\delta,\bar\alpha)}\]
with probability more than $1-\delta_0$. Here $\Pi_P$ represents an optimal solution of the non-robust multiperiod model under $P$.
\end{prob}
In order to use \cite{blanchet2021distributionally}, we define $\Pi_P$ as the set of all solutions $A$ of 
\begin{equation}
	\label{non-robust}
	\begin{aligned}
		\min_{A}\; A^\intercal E_P[MM^\intercal]A\quad \text{ subject to }\quad &E_P\edg{A^\intercal M}=\lambda\text{ and }\sum_{i=1}^n(A^i_t)^1=1 \text{ for each } t.
	\end{aligned}
\end{equation} 
The $(A^i_t)^1:=f^i_t$ represents the constant term of the Taylor series of $\pi^i_t$: see \eqref{eq7} and \eqref{eq8_1}.
\begin{remark}
    The condition $\sum_{i=1}^n(A^i_t)^1=1$ is an approximation for the condition $\sum_{i=1}^n\pi_t^i=\sum_{i=1}^n(A^i_t)^\intercal N_t^i=1$ with Taylor series approximation of $0^{th}$ order. Such approximation enables us to obtain $\lambda^*_0$ in the first step of RWPI protocol in Section \ref{specification}.
\end{remark}



Throughout this article, we assume the following technical conditions:
	\begin{assumption}
		\label{as1}
		The problem \eqref{non-robust} has a unique solution $A$ under any $P$, that is $\Pi_{P}=\crl{A}$.
	\end{assumption}
	\begin{assumption}
		\label{as2}
		$\boldsymbol{R}$ has a probability density function and the underlying return time series $\crl{\boldsymbol{R}_t:t\geq0}$ is a stationary, ergodic process satisfying $E_{P^*}(\|\boldsymbol{R}\|^{2T}_{2T})<\infty$ for each $t\geq0$.
	\end{assumption}
	\begin{assumption}
		\label{as3}
		For every measurable function $h(\cdot)$, we have
		\eqn{N^{1/2}[E_{Q}(h(\boldsymbol{R}))- E_{P^*}(h(\boldsymbol{R}))]\Rightarrow N(\boldsymbol{0},Cov[h(\boldsymbol{R})])}
		as $N\to\infty$, where
		\eqn{Cov[h(\boldsymbol{R})]:=\lim_{N \to \infty}Var_{P^*}\Big(N^{-1/2}\sum^N_{i=1}h(\boldsymbol{R}^{(i)})\Big)}
		satisfying $|h(w)|\leq k(1+\|w\|^2_2)$ for some $k>0$. Here, the $N(\cdot,\cdot)$ represents the normal distribution under $P^*$.
	\end{assumption}

	\begin{assumption}
		\label{as4}
For any matrix $\hat\alpha\in\bbR^{k\times k}$ and a vector $\hat\beta\in\bbR^k$ such that $\hat\alpha \neq \boldsymbol{0}$ or $\hat\beta\neq \boldsymbol{0}$, 
\eqn{P^*(\|\hat\alpha \boldsymbol{R}+ \hat\beta\|_2>0)=1.}
	\end{assumption}
	
		\begin{assumption}
		\label{as5}
		Let $S_t$ be the price of an asset at time point $t$ in the $L^2$ space. We define our return $R_t=(S_t-S_{t-\Delta t})/S_{t-\Delta t}$ on the time period $(t-\Delta t, t]$. When $\Delta t$ converges to $0$, the return $R_t$ also converges to $0$. 
	\end{assumption}
The following proposition enables us to apply the single-period result of \cite{blanchet2021distributionally} to our problem.
	\begin{proposition}
		\label{Remark 2}
		Assumption \ref{as2}, \ref{as3}, and \ref{as4} also hold for $M$ as well as $\boldsymbol{R}$.
	\end{proposition}
	\begin{proof}
We know that there exists the transformation function $g(\cdot)$ from $\boldsymbol{R}$ to $M$ for $s$-period satisfying $M_t = g(\boldsymbol{R}_{ts-s+1},...,\boldsymbol{R}_{ts})$. According to Theorem 3.35 of \cite{white2014asymptotic}, we know if $\crl{\boldsymbol{R}_t}$ is stationary and ergodic process, then $\crl{M_t}$ is also stationary and ergodic process satisfying the property $E_{P^*}(\|\boldsymbol{M}\|^2_2)<\infty$.

On the other hand, Assumption \ref{as3} also holds for $M$ if we let $h=\tilde h\circ g$ in Assumption \ref{as3}, where
$\tilde h$ is an arbitrary measurable function.

Assumption \ref{as4} also holds for $M$. For $k$-dimensional $M$, any matrix $\hat\alpha\in\bbR^{k\times k}$ and a vector $\hat\beta\in\bbR^k$ such that $\hat\alpha \neq \boldsymbol{0}$ or $\hat\beta\neq \boldsymbol{0}$, we have
\eqn{P^*(\|\hat\alpha M+ \hat\beta\|_2>0)=1.}
Since the PDF of $M$ can be derived from the PDF of $\mathbf{R}$,
\eqn{P^*(\hat\alpha x+\hat\beta=\boldsymbol{0})=0,}
for any given $\hat\alpha\in \mathbb{R}^{k\times k}\backslash\{\boldsymbol{0}\},\hat\beta\in\mathbb{R}^k\backslash\{\boldsymbol{0}\}$, and any $x\in\mathbb{R}^k$. This implies that
\eqn{P^*(\|\hat\alpha M+ \hat\beta\|_2=0)=0.}
This completes the proof.
\end{proof}

\subsection[Specification of the ambiguity set]{Specification of $\delta$ and $\bar\alpha$}\label{specification}
We choose $\delta$ as the minimum uncertainty level so that the optimal solution of the non-robust multiperiod model under $P^*$ is in the plausible estimate of $A^*$ with confidence level $1-\delta_0$ which the investor assigns. In other words, we define
	\begin{align*}
		\delta :=\min\crl{\varepsilon:P^*\left(A^*\in\bigcup_{P\in \scU_{\varepsilon}(Q)}\Pi_P\right)\geq 1-\delta_0}
	\end{align*}
where $\Pi_P$ is the set of all solutions of the non-robust multiperiod portfolio selection problem (\ref{non-robust}).

Likewise, for given $\delta$, we choose $\bar\alpha$ as the maximum value that we do not rule out $A^*$, that is
\begin{align*}
	\bar\alpha:=\max\crl{\alpha: P^*\brak{A^*\in \scF_{(\delta,\bar{\alpha})}}\geq 1-\delta_0 }.
\end{align*}

In practice, one can use the framework illustrated in \cite{blanchet2021distributionally} to find $\delta$ and $\bar\alpha$ for given data. For readers' convenience, we will describe it here.
\begin{algo}
	\begin{enumerate}
		\item Find $(A^*,\lambda_0^*)$ by solving \eqref{non-robust} for $P=Q$ using the Lagrange multiplier method:
\eq{\label{eq2.3}&2\Sigma_NA^*-\lambda_0^*\mu_N-\brak{\sum_{t=1}^T\lambda_t^* \boldsymbol{\tilde 1}_t}=\boldsymbol{0},\\
	&(A^*)^\intercal \mu_N-\lambda = 0,\\
	&\sum_{i=1}^n(A^{*,i}_t)^1 =1\text{ for all }t=1,2,...,T}
where $(\mu_N,\Sigma_N):=(E_{Q}[M],E_{Q}[MM^\intercal])$ and $\boldsymbol{\tilde 1}_t:=(\partial_A\brak{\sum_{i=1}^n(A^i_t)^1}^\intercal$.
		\item Let $h(x)=x+2(\lambda_0^*)^{-1}xx^\intercal A^*$. For a random variable $Z\sim N(0, \Upsilon_h)$ where $$\Upsilon_h:=\lim_{N\to\infty}Var_{P^*}(N^{-1/2}\Sigma_{k=1}^Nh(M_k)),$$ we let
		\[
		\delta^*:=\min\crl{\delta:P^*\left(\frac{\|Z\|^2_2}{4(1-\mu^\intercal_N\Sigma^{-1}_N\mu_N)}\leq\delta N\right)\geq 1-\delta_0}.
		\]
		\item Let $(m_j)_{j\geq 1}$ be the realizations of $M$. For a random variable $\tilde Z\sim N(0,\Upsilon_{A^*})$ where
		\[		\Upsilon_{A^*}:=\lim_{\bar N\to\infty}Var_{P^*}\left(\bar N^{-1/2}\sum_{j=1}^{\bar N}(A^*)^\intercal m_j/\|A^*\|_2\right),
		\]
		let $\bar \alpha^*:=\lambda -\sqrt{\delta}\norm{A^*}_2\max\crl{s_0,s_0'}$ where
		\begin{align*}
  s_0&:=\min\crl{s:P^*\brak{\tilde Z\geq\sqrt{\delta N}(1-s_0)}\geq 1-\delta_0}\\
			s'_0&:=\frac{\lambda+\sqrt{\delta}\norm{A^*}_2-(A^*)^\intercal \mu_N}{\sqrt{\delta}\norm{A^*}_2}
		\end{align*}
  \item $(\delta^*, \bar\alpha^*)$ is the proposed value.
	\end{enumerate}
\end{algo}
\begin{Remark}
Step 2 is proved for the case where $(\mu_*,\sigma_*):=(E_P[M],E[MM^\intercal])$ and $\bar \lambda_0$ defined in the proof instead of $(\mu_N,\Sigma_N)$ and $\lambda_0^*$, but we use the empirical estimates to obtain $\delta^*$. In Step 3, $\Upsilon_{A^*}$ is obtained asymptotically.
\end{Remark}
\begin{Remark} When we obtain the proposed $\bar\alpha$ in Step 3, we ignore the condition $\sum_{i=1}^n(A^i_t)^\intercal N_t^i=1$ in the restriction of the feasible set \eqref{feasible}. However, we impose the condition when we perform the simulations in Section \ref{section5} to Section \ref{section7}.
\end{Remark}

\begin{proof} We don't need to prove Step 1 and Step 4. 
First, let us prove Step 3 assuming that we have obtained $\delta^*$ in Step 2.\\
\noindent (Proof of Step 3) Instead of finding $\bar\alpha^*$ directly, we will find the appropriate value for
$s_0=\displaystyle\frac{1}{\sqrt{\delta}\|A^*\|_2}(\lambda-\Bar{\alpha})$. Since 
\eqn{\min_{P\in \scU_\delta(Q)} E_P[(A^*)^{\intercal}M]=(A^*)^\intercal \mu_N-\sqrt{\delta}\|A^*\|_2\geq\Bar{\alpha}}
(see Section \ref{section4}) and $\lambda=(A^*)^\intercal E_{P^*}[M]$, we have
\eqn{(A^*)^\intercal(\mu_N-E_{P^*}(M))\geq\|A^*\|_2\sqrt{\delta}(1-s_0).}
Due to Proposition 1 in \cite{blanchet2021distributionally}, we know 
\eqn{N^{1/2}\left[\frac{(A^*)^\intercal(\mu_N-E_{P^*}(M))}{\|A^*\|_2}\right]\Rightarrow \tilde Z \sim N(0,\Upsilon_{A^*}),}
where
\eqn{\Upsilon_{A^*}:=\lim_{N\to\infty}Var_{P^*}\left({ N}^{-1/2}\sum_{j=1}^{ N}(A^*)^\intercal m_j/\|A^*\|_2\right).}
Now one can  choose $s_0$ such that the following inequality \eqref{s_0} holds with the confident level $1-\delta_0$,
\eq{\label{s_0}\tilde Z
	\geq\sqrt{\delta N}(1-s_0).}
Since 
\[
\frac{(A^*)^\intercal\mu_N-\lambda}{\|A^*\|_2}\geq \sqrt{\delta}(1-s_0)
\]
may not always hold, we define $s'_0$ to be the value
\[
\frac{(A^*)^\intercal\mu_N-\lambda}{\|A^*\|_2}= \sqrt{\delta}(1-s'_0)
\]
and let $s=\max\crl{s_0,s_0'}$. For such $s$, our estimate for $\bar\alpha$ becomes $\lambda -\sqrt{\delta}\norm{A^*}_2s.$

\noindent (Proof of Step 2) Let $(\mu_*,\Sigma_*)=(E_{P^*}[M],E_{P^*}[MM^\intercal])$. Assume that we have the solution of \eqref{non-robust} when $P=P^*$ and let $(\bar A, \bar \lambda_0, (\bar \lambda_t)_{t=1,2,..,T})$ be its corresponding Lagrange equation given as
\eq{\label{asdi}&2\Sigma_* \bar A-\bar\lambda_0\mu_*-\sum_{t=1}^T\bar\lambda_t \boldsymbol{\tilde 1}_t=0,\\
	&\bar A^\intercal \mu_*-\lambda = 0,\\
	&\sum_{i=1}^n(\bar A^{i}_t)^1 =1\text{ for all }t=1,2,...,T}
 Here, the well-definedness is guaranteed by the Assumption \ref{as1}.
Let
\[
\bar\mu =\frac{1}{\lambda_0}\brak{2\Sigma_N \bar A-\sum_{t=1}^T\bar\lambda_t \boldsymbol{\tilde 1}_t}
\]
Following the idea of RWPI approach introduced in \cite{blanchet2019robust} and Section 4.1 of \cite{blanchet2021distributionally}, let the RWP function be
\eqn{\scR_N(\Sigma_N,\bar\mu)=\inf_{P}\crl{W_2(P,Q):\begin{aligned}E_P[MM^\intercal]&=\Sigma_N,E_P[M]=\bar\mu\end{aligned} }}
where $\bar \mu=h(\Sigma_N)$.
Then, our choice for $\delta$ is given by
\eqn{
	\delta^* :=\inf\crl{\delta:P^*\left(\scR_N(\Sigma_N,\bar \mu)\leq\delta\right)\geq 1-\delta_0}.
}
Note that, by the argument in the proof of Theorem 2 in \cite{blanchet2021distributionally}, we have
\[
N\scR_N(\Sigma_N,\mu_N)\Rightarrow \frac{\norm{Z}_2^2}{4(1-\mu_*^\intercal\Sigma_*^{-1}\mu_*)}\quad\text{ as }N\to\infty
\]
where $Z$ is the distributional limit of $N^{-1/2}\sum_{i=1}^N\brak{m_i-\bar \mu}$. 
Note that, by multiplying \eqref{asdi} by $\bar A^\intercal$ and substituting $\lambda$ for $\bar A^\intercal\mu_*$ and $\bar A^\intercal\bar\mu$, we obtain the following equations
\begin{align*}
2\Sigma_*  \bar A&-\bar\lambda_0\mu_*-\sum_{t=1}^T\bar\lambda_t \boldsymbol{\tilde 1}_t=0, &2\bar A^\intercal \Sigma_* \bar A&-\bar\lambda_0\lambda-\bar A^\intercal\sum_{t=1}^T\bar\lambda_t \boldsymbol{\tilde 1}_t  =0\\
2\Sigma_N  \bar A&-\bar\lambda_0\bar\mu-\sum_{t=1}^T\bar\lambda_t \boldsymbol{\tilde 1}_t=0, &2\bar A^\intercal \Sigma_N \bar A&-\bar\lambda_0\lambda-\bar A^\intercal\sum_{t=1}^T\bar\lambda_t \boldsymbol{\tilde 1}_t =0
\end{align*}
for some $\bar\lambda_t$ and $\lambda_t^*$ for $t=1,2,...,T$.
Therefore, if we let $H_N:=\Sigma_N-\Sigma_*$, we have
\begin{align*}
    \bar\mu-\mu_* = \frac{2}{\bar \lambda_0}H_N\bar A.
\end{align*}
By Proposition \ref{Remark 2}, for $h_1(x)=x$ and $h_2(x)=xx^T$, we have
\begin{align}
    N^{-1/2}\sum_{i=1}^{N}\brak{m_i-\mu_*}&\Rightarrow N(0,\Upsilon_{h_1})\\
    N^{1/2} H_N&\Rightarrow N(0,\Upsilon_{h_2})
\end{align}
Since
\[\begin{aligned}
N^{-1/2}\sum_{i=1}^N\brak{m_i-\bar \mu}&= N^{-1/2}\sum_{i=1}^N\brak{m_i-\mu_*}+N^{1/2}\brak{\mu_*-\bar \mu}\sim Z_0+\frac{2}{\bar \lambda_0}Z_1\bar A
\end{aligned}
\]
for some $Z_0\sim N(0,\Upsilon_{h_1})$ and $Z_1\sim N(0,\Upsilon_{h_2})$. Therefore, $Z$ can be chosen so that
\[
Z\sim N(0,\Upsilon_h)
\]
with $h(x)=x+2\bar\lambda_0^{-1}xx^\intercal \bar A$. Since $\bar\lambda_0$ and $\bar A$ can be approximated by $\lambda_0^*$ and $A^*$, we proved the Step 2.
\end{proof}

\section{Approximation for Investment Strategy}
\label{section3}
From our definition \eqref{feasible} of feasible region, there exist measurable functions $f^i$
\begin{align}
    \label{eq6}\pi_t^i=f^i(T-t,R_1,R_2,...,R_{t-1},R_{t},...,R_T),
\end{align}
which are analytic with respect to $\mathbf{R}$. Since $f^i$ are analytic, we can use arbitrary order Taylor series to approximate $f^i$. Using Taylor series, we can transform the multi-period model to single period model. 

Without losing generality, we will provide the framework for the second order Taylor series approximation. Since $\pi^i$ should be predictable, when $b\geq t$ or $d\geq t$, we let the components $\partial f^i(T-t,0,...,0)/\partial x_d=0$ and $\partial^2 f^i(T-t,0,...,0)/\partial x_d\partial x_b=0$. 
Therefore,
\begin{equation}\label{eq7}
\begin{aligned}
	\pi_t^i
	& \approx f^i(T-t,0,...,0)+\frac{\partial}{\partial x^c_d}f^i(T-t,0,...,0)\cdot R^c_d \indicator{d\leq t-1}\\
	&\quad+\frac{1}{2}\frac{\partial^2}{\partial x^a_b\partial x^c_d}f^i(T-t,0,...,0)\cdot R^a_bR^c_d \indicator{b,d\leq t-1}.
\end{aligned}
\end{equation}
Here, for notational brevity, we use Einstein notation.\footnote{When an index variable appears more than twice in a single term, it implies the summation of that term over all the values of the index.} Then, the final wealth can be expressed as
\begin{equation}\label{eq8}
\begin{aligned}
\boldsymbol{\pi}^{\intercal}\boldsymbol{R}&=
\pi_1^1R_1^1+...+\pi_1^nR_1^n+...+\pi_t^iR_t^i+...+\pi_T^1R_T^1+...+\pi_T^nR_T^n\\
&=f^i(T-t,0,...,0)\cdot R_t^i +\frac{\partial}{\partial x^c_d}f^i(T-t,0,...,0)\cdot R^c_dR_t^i\indicator{d\leq t-1}\\
&\quad+\frac{1}{2}\frac{\partial^2}{\partial x^a_b\partial x^c_d}f^i(T-t,0,...,0)\cdot R^a_bR^c_dR_t^i\indicator{b,d\leq t-1}\\
&=f_t^i\cdot R_t^i+g_{dt}^{ci}\cdot R_{d}^cR_{t}^i+h_{bdt}^{aci}\cdot R_b^aR_d^cR_t^i
\end{aligned}
\end{equation}
where \eq{\label{eq8_1}&f_t^i=f^i(T-t,0,...,0),\quad g_{dt}^{ci}=\frac{\partial}{\partial x_d^c}f^i(T-t,0,...,0)\indicator{d\leq t-1},\\ &h_{bdt}^{aci}=\frac{1}{2}\frac{\partial^2}{\partial x^a_b\partial x^c_d}f^i(T-t,0,...,0)\indicator{b,d\leq t-1}}
This indicates a summation with the index running through the integral numbers 
$a,c,i=1,...,n$ and $b,d,t=1,...,T$ in the order of `$abcdit$'.

Next, we try to transform the formulation of summation (\ref{eq8}) to the formulation of vector product
\eq{\label{eq9}\boldsymbol{\pi}^{\intercal}\boldsymbol{R}&=f_t^i\cdot R_t^i+g_{dt}^{ci}\cdot R_{d}^cR_{t}^i+h_{bdt}^{aci}\cdot R_b^aR_d^cR_t^i\\
&=f_1^1\cdot R_1^1+\cdots+f_T^n\cdot R_T^n+g_{1,1}^{1,1}\cdot R_1^1R_1^1+\cdots+g_{T,T}^{n,n}\cdot R_T^nR_T^n\\
&\quad+h_{1,1,1}^{1,1,1}\cdot R_1^1R_1^1R_1^1+\cdots+h_{T,T,T}^{n,n,n}\cdot R_T^nR_T^nR_T^n\\
&=\Big(f_t^i,g_{dt}^{ci},h_{bdt}^{aci}\Big)^\intercal\Big(R_t^i,R_{d}^cR_{t}^i,R_b^aR_d^cR_t^i\Big),}
where\eqn{&(f_t^i)=(f_1^1,f_1^2,...,f_T^n), (g_{dt}^{ci})=(g_{11}^{11},g_{11}^{21},...,g_{TT}^{nn}), (h_{bdt}^{aci})=(h_{111}^{111},h_{111}^{211},...,h_{TTT}^{nnn}),
(R_t^i)=(R_1^1,R_1^2,...,R_T^n),\\&(R_{d}^cR_{t}^i)=(R_1^1R_1^1,R_1^2R_1^1,..., R_T^nR_T^n),(R_b^aR_d^cR_t^i)=(R_1^1R_1^1R_1^1,R_1^2R_1^1R_1^1,...,R_T^nR_T^nR_T^n).}
{\color{black}Here, the $p$-th component of a vector $(v^{lmn}_{opq})_{l,m,n,o,p,q}$ is $v^{aci}_{bdt}$ where,
\eqn{&t=\left\lceil\frac{p}{T^2n^3}\right\rceil, i=\left\lceil\frac{p-(t-1)n^3T^2}{n^2T^2}\right\rceil,d=\left\lceil\frac{p-(t-1)n^3T^2-(i-1)n^2T^2}{n^2T}\right\rceil,\\
&c=\left\lceil\frac{p-(t-1)n^3T^2-(i-1)n^2T^2-(d-1)n^2T}{nT}\right\rceil,\\
&b=\left\lceil\frac{p-(t-1)n^3T^2-(i-1)n^2T^2-(d-1)n^2T-(c-1)nT}{n}\right\rceil\\
&a= p-(t-1)n^3T^2-(i-1)n^2T^2-(d-1)n^2T-(c-1)nT-(b-1)n
}
where $\lceil x \rceil:=\min\crl{z\in\bbZ: x\leq z}$.
}

Let vector $A$ and vector $M$ be
\eqn{A=(f_t^i,g_{dt}^{ci},h_{bdt}^{aci}),\;M=( R_t^i,R_{d}^cR_{t}^i,R_b^aR_d^cR_t^i).}
Therefore, our optimization problem \eqref{eq1} becomes
\begin{equation}\label{eq10}
	\begin{aligned}
	\inf_{A\in\scF_{(\delta,\Bar{\alpha})}}\max_{\lambda\geq \bar\alpha}\brak{\sup_{P\in \scU_\delta(Q),E_P[A^{\intercal}M]=\lambda}E_P\edg{\brak{A^\intercal M}^2-\lambda^2}}
	\end{aligned}
\end{equation}
Fixing $E_P[A^{\intercal}M]=\lambda$ is helpful because the inner maximization problem is now linear in $P$. 
\begin{theorem}
The primal optimization problem given in (\ref{eq1}) is equivalent to the following dual problem
\eq{\inf_{A\in\scF_{(\delta,\Bar{\alpha})}}\sqrt{A^\intercal Var_Q(M)A}+\sqrt{\delta}\|A\|_2.\label{eq15}}
\end{theorem}
\begin{proof}
Let's focus on solving the following problem
\eq{\sup_{P\in \scU_\delta(Q),E_P[A^{\intercal}M]=\lambda}E_P[(A^{\intercal}M)^2]\label{eq11}}
Let \begin{align}
l(A,\lambda):=E_Q[(A^\intercal M)^2]&+2(\lambda-E_Q[A^\intercal M])E_Q[A^\intercal M]+\delta \|A\|_2^2\\
&+2\sqrt{\delta\|A\|_2^2-(\lambda-E_Q[A^\intercal M])^2}\sqrt{A^\intercal Var_Q(M)A}.
\end{align}
Then, Proposition A.3 \cite{blanchet2021distributionally} tells you that, if $\delta\|A\|_2^2-(\lambda-E_Q[A^\intercal M])^2\geq0$, the value of (\ref{eq11}) is equal to $l(A,\lambda)$. According to the previous result, we have
\eqn{l(A,\lambda)-\lambda^2&=E_Q[(A^\intercal M)^2]-(E_Q[A^\intercal M])^2+\crl{\delta\|A\|_2^2-(\lambda-E_Q[A^\intercal M])^2}\\
&\quad+2\sqrt{\delta\|A\|_2^2-(\lambda-E_Q[A^\intercal M])^2}\sqrt{A^\intercal Var_Q(M)A}\\
&=A^\intercal Var_Q(M)A+\crl{\delta\|A\|_2^2-(\lambda-E_Q[A^\intercal M])^2}\\
&\quad+2\sqrt{\delta\|A\|_2^2-(\lambda-E_Q[A^\intercal M])^2}\sqrt{A^\intercal Var_Q(M)A}\\
&=\Big(\sqrt{A^\intercal Var_Q(M)A}+\sqrt{\delta\|A\|_2^2-(\lambda-E_Q[A^\intercal M])^2}\Big)^2.}
So the middle optimization problem in (\ref{eq10}) becomes
\eq{
	\max_{\lambda\geq \bar\alpha}\brak{\sup_{P\in \scU_\delta(Q),E_P[\boldsymbol{\pi}^{\intercal}\boldsymbol{R}]=\lambda}E_P\edg{\brak{\boldsymbol{\pi}^\intercal\boldsymbol{R}}^2-\lambda^2}}
	&=\max_{\lambda\geq \bar\alpha,\delta\|A\|_2^2-(\lambda-E_Q[A^\intercal M])^2\geq 0}[l(A,\lambda)-\lambda^2]
	\\&=\Big(\sqrt{A^\intercal Var_Q(M)A}+\sqrt{\delta}\|A\|_2\Big)^2.
}
\end{proof}
Note the optimal solution $A$ has the restriction that some components are $0$, but it does not change the value of our optimization problem (\ref{eq15}): the proof in the Appendix \ref{appendix}.

\section{Feasible Region}
\label{section4}
In the previous section, we obtain our robust mean-variance optimization dual problem. Before solving the optimal solution $A$ in (\ref{eq15}), we need to present the feasible region for the problem.

Firstly, we notice that the ambiguity set $\delta$ and the worst mean return target $\bar{\alpha}$ in the feasible region $ \mathcal{F}_{(\delta,\Bar{\alpha})}$ need to be chosen very carefully. On the one hand, if $\delta$ is chosen too large, the model ambiguity becomes to much and the solution would be too conservative. On the other hand, if $\delta$ is chosen too small, the effect of robustness will be too small to reflect the difference between the underlying and historical probability measures. Once $\delta$ has been selected, choosing $\bar{\alpha}$ is crucial, which needs to consider the size of $\delta$. It can protect the optimal solutions from being too aggressive and make the model more sensible. About the guidance on the choice of the size of $\delta$ and $\bar{\alpha}$, we refer to Blanchet and Chen (2018).

Note that the optimal solution $A$ in the feasible region $ \mathcal{F}_{(\delta,\Bar{\alpha})}$ has some conditions. Firstly, we assume the investment of risky assets in each period is $\sum_{i=1}^n\pi^i_t = 1$. In order to do the numerical simulation, we know that $\pi_t^i$ can be recovered from the data. We want to change $\pi_t^i$ to the formulation of vector product
\eqn{\pi^i_t\approx (A_t^i)^\intercal N_t^i.
}

In our paper, we previously use the second order Taylor expansion to $\pi^i_t$ and have
\eqn{\pi_t^i
	& \approx f^i(T-t,0,...,0)+\frac{\partial}{\partial x^c_d}f^i(T-t,0,...,0)\cdot R^c_d \indicator{d\leq t-1}\\
	&\quad+\frac{1}{2}\frac{\partial^2}{\partial x^a_b\partial x^c_d}f^i(T-t,0,...,0)\cdot R^a_bR^c_d \indicator{b,d\leq t-1}\\
	&=(f_t^i,g_{dt}^{ci},h_{bdt}^{aci})^\intercal (1,R_{d}^c,R_b^aR_d^c).}
By letting $A_t^i=(f_t^i,g_{dt}^{ci},h_{bdt}^{aci})$ and $N_t^i=(1,R_{d}^c,R_b^aR_d^c),$ we have
\eq{\sum^n_{i=1}\pi^i_t\approx \sum^n_{i=1}(A_t^i)^\intercal N_t^i=1.}
Next, note that $A$ in the feasible region $ \mathcal{F}_{(\delta,\Bar{\alpha})}$ has the other  constraint, that is
\[\min_{P\in\scU_\delta(Q)}[E_P(A^{\intercal}M)]\geq\Bar{\alpha} \iff -\max_{P\in \scU_\delta(Q)}[E_P((-A)^{\intercal}M)]\geq \Bar{\alpha}.\]
By the dual formulation in Proposition 1 of \cite{blanchet2019robust},
we have
\eq{\label{eq19}\max_{P\in \scU_\delta(Q)}[E_P((-A)^{\intercal}M)]=\inf_{\gamma\geq0}[\gamma\delta+\frac{1}{n}\sum_{i=1}^n\Phi_\gamma(m_i)]}
where \eq{\label{eq20}\Phi_\gamma(m_i)&=\sup_u\{(-A)^{\intercal}u-\gamma\|u-m_i\|_2^2\}\\
&=\sup_\Theta\{(-A)^{\intercal}(\Theta+m_i)-\gamma\|\Theta\|_2^2\}=\sup_\Theta\{(-A)^{\intercal}\Theta-\gamma\|\Theta\|_2^2\}-A^\intercal m_i\\
&=\sup_\Theta\{\|A\|_2\|\Theta\|_2-\gamma\|\Theta\|_2^2\}-A^{\intercal}m_i\\
&=\frac{\|A\|_2^2}{4 \gamma}-A^\intercal m_i}
Then, (\ref{eq19}) becomes
\eq{\label{eq21}\max_{P\in \scU_\delta(Q)}[E_P((-A)^{\intercal}M)]&=\inf_{\gamma\geq0}\{\gamma\delta+\frac{1}{n}\sum_{i=1}^n[\frac{\|A\|_2^2}{4\gamma}-A^{\intercal}m_i]\}\\
&=\inf_{\gamma\geq0}\{\gamma\delta+\frac{\|A\|_2^2}{4\gamma}-E_{Q}[A^{\intercal}M]\}\\
&=\sqrt{\delta}\|A\|_2-E_{Q}[A^{\intercal}M]}
or\eq{\label{eq22}\min_{P\in \scU_\delta(Q)}[E_P(A^{\intercal}M)]=E_{Q}[A^{\intercal}M]-\sqrt{\delta}\|A\|_2}
Therefore, we can have our feasible region
\eq{\mathcal{F}_{(\delta,\Bar{\alpha})}:=\crl{A:E_{Q}[A^{\intercal}M]-\sqrt{\delta}\|A\|_2\geq\Bar{\alpha}, \sum^n_{i=1}(A_t^i)^\intercal N_t^i=1}.}

\section{Implementation of Multi-period Investment Strategy}
\label{section5}
When one implements a multi-period investment strategy, rebalancing of the portfolio becomes no longer trivial. Assume that we have $T$-period strategies $$(\pi^{strategy}_t:=f(T-t, R_1,R_2,...,R_T))_{t=1,2,...,T}$$ where $f(T-t,r_1,r_2,...,r_T)$ only depends on the first $t$ variables $(T-t, r_1,r_2,...,r_{t-1})$. We consider the backtesting of the strategy on the time $\brak{s\in\bbZ:0\leq s\leq N}$ where $N\gg T$. A naive implementation would be, for $s,t\in\bbN, k\in\bbZ$ such that $s=kT+t$ and $1\leq t\leq T$,
\begin{align*}
	\pi^{naive}_s=f(T-t, R_{kT+1},R_{kT+2},...,R_{(k+1)T})
\end{align*}
for all $s\in[0,N]$. However, this is not a good rebalancing strategy because, for example, at time $T+1$, the strategy does not use any information from time $2$ to $T$. In addition, in our Taylor series approximation scheme, if there are differences between $f(T-t,0,0,...,0), t\in[1,T]$, we will need to rebalance our portfolio even when the stock prices are constant. This will result in huge transaction fee by default.

In order to circumvent these problem, we equally divide our portfolio into $T$ number of subportfolios which start at different times and follow the same naive strategy. 

For $i\in\crl{1,2,...,T}, s,t\in\bbN, k\in\bbN\cup\crl{0}$ such that $s=kT+t+(i-1)$ and $1\leq t\leq T$, we define the target investment strategy
\begin{align*}
	\pi_s:=\frac{1}{T}\sum_{i=1}^Tf(T-t, R_{kT+i},R_{kT+i+1},...,R_{(k+1)T+i-1})\indicator{i\leq s}.
\end{align*}
An example for $T=3$ is illustrated in Table \ref{tab:target}. 
\begin{table}[H]
	\begin{tabular}{|c|c|c|c|c|}
		\hline
		Time $(t-1,t]$& Subportfolio $\pi^{(1)}$      & Subportfolio $\pi^{(2)}$       & Subportfolio $\pi^{(3)}$& Target Investment $\pi$   \\ \hline
		$t=1$ & $f(2, R_1,R_2, R_3)$ & $0$                   & $0$&  $\frac{1}{3}\sum_{l=1}^3\pi^{(l)}_1$\\
		$t=2$ & $f(1, R_1,R_2, R_3)$ & $f(2,R_2,R_3,R_4)$    & $0$&  $\frac{1}{3}\sum_{l=1}^3\pi^{(l)}_2$              \\
		$t=3$ & $f(0, R_1,R_2, R_3)$ & $f(1,R_2,R_3,R_4)$    & $f(2,R_3,R_4,R_5)$& $\frac{1}{3}\sum_{l=1}^3\pi^{(l)}_3$\\
		$t=4$ & $f(2, R_4,R_5, R_6)$ & $f(0,R_2,R_3,R_4)$    & $f(1,R_3,R_4,R_5)$& $\frac{1}{3}\sum_{l=1}^3\pi^{(l)}_4$ \\
		$t=5$ & $f(1, R_4,R_5, R_6)$ & $f(2,R_5,R_6,R_7)$    & $f(0,R_3,R_4,R_5)$&$\frac{1}{3}\sum_{l=1}^3\pi^{(l)}_5$ \\
		$t=6$ & $f(0, R_4,R_5, R_6)$ & $f(1,R_5,R_6,R_7)$    & $f(2,R_6,R_7,R_8)$& $\frac{1}{3}\sum_{l=1}^3\pi^{(l)}_6$\\
		$t=7$ & $f(2, R_7,R_8, R_9)$ & $f(0,R_5,R_6,R_7)$    & $f(1,R_6,R_7,R_8)$&$\frac{1}{3}\sum_{l=1}^3\pi^{(l)}_7$ \\ 
		$t=8$ & $f(1, R_7,R_8, R_9)$ & $f(2,R_8,R_9,R_{10})$ & $f(0,R_6,R_7,R_8)$& $\frac{1}{3}\sum_{l=1}^3\pi^{(l)}_8$\\ \hline
	\end{tabular}
	\caption{Target investment for 3-period model}
	\label{tab:target}
\end{table}

In practice, it is not plausible to rebalance portfolio daily because of the transaction costs. Therefore, we implement the strategy in the following way.
At time $0$, we invest $\pi_1^*:= \frac{1}{T}\sum_{k=1}^{T} f(k-1,0,...,0)$ and
let it evolve until
\begin{align}\label{timetorebalance}
	\max_{i=1,...,n}\left|\displaystyle\frac{(\pi^i_1)^*(1+R^i_1)(1+R^i_2)\cdots(1+R^i_{t-1})-\pi^i_t}{(\pi^i_1)^*(1+R^i_1)(1+R^i_2)\cdots(1+R^i_{t-1})}\right|> 0.05.
\end{align}
When \eqref{timetorebalance} happens, we rebalance our portfolio to align with our target investment strategy $\pi$. Again, we wait until the $\pi^*$ deviates from $\pi$ more than 5\% and if it does, we rebalance our portfolio. We repeat the scheme until $N$. See the following algorithm.
\begin{itemize}
\item[]{\bf Algorithm}
	\begin{itemize}
		\item[(i)] At time $0$, we let $\pi_1^*:=\frac{1}{T} \sum_{k=1}^{T} f(k-1,0,...,0)$ and $\tau_0=1$.
		\item[(ii)] For $k\in\bbN$, define stopping times
		\begin{align*}
			\tau_k:=\min\crl{\tau_{k-1}<t\leq T:\max_{i=1,...,n}\left|\displaystyle\frac{(\pi^i_{\tau_{k-1}})^*(1+R^i_1)(1+R^i_2)\cdots(1+R^i_{t-1})-\pi^i_t}{(\pi^i_1)^*(1+R^i_1)(1+R^i_2)\cdots(1+R^i_{t-1})}\right|> 0.05},
		\end{align*}
	and let
	\begin{align*}
		(\pi^i_t)^*:=
		\begin{cases}
			(\pi^i_{\tau_{k-1}})^*(1+R^i_{\tau_{k-1}})(1+R^i_{\tau_{k-1}+1})\cdots(1+R^i_{t-1}),& \text{if } t\in(\tau_{k-1},\tau_k)\\
			\pi_{\tau_k} &\text{if } t= \tau_k
		\end{cases}.
	\end{align*}
	\item[(iii)] We let $(\pi^*_t)_{t=1,2,...,N}$ be our investment strategy.
	\end{itemize}
\end{itemize}
\begin{Remark}
	Note that the transaction costs of the above algorithm for time $[1,T]$ is not optimised. However, since $N\gg T$, the investments $\pi^*_t$ for $t<T$ have virtually no effect on the performance of investment overall. 
\end{Remark}

When we assume there exist transaction costs, we use the linear transaction costs of 0.2\%. More precisely, our transaction costs are given by
\begin{align*}
	TC=0.002\sum_{k\in\bbN,\tau_k\leq T}\sum^n_{i=1}\left|(\pi^i_{\tau_k-1})^*(1+R_{\tau_k-1})-(\pi^i_{\tau_k})^*\right|.
\end{align*}

\section{Empirical Studies}
\label{section6}
In this section, we consider the $L^2$-norm for our robust optimization model $(\ref{eq15})$. We use the first-order Taylor expansion in \eqref{eq7} for our investment strategy\footnote{According to our simulation which is not included in this article, the second-order Taylor expansion model produces essentially the same performance as the first order Taylor expansion model.} and compare its performance with those of well-known strategies.

The data provided by Bloomberg Terminal is used for the empirical study. We select 100 largest S\&P 500 companies by market cap at five different time points, Feb 1 2002, Jun 1 2004, Jun 1 2006, Aug 1 2008, and June 1 2009. For each time point, we use the daily adjusted closing price data for the previous 2 years to estimate our parameters and simulate the strategies for the following 8 years. We assume 
the net investment of each period is a fixed constant. We compare the means, the standard deviations, and the Sharpe ratio of wealth returns of our robust 2-period model and other conventional methods with and without the transaction costs.

{\color{black}We choose Feb 1 2002 as the first start date because the 9/11 attacks happened in 2001. Start dates Jun 1 2004 and Jun 1 2006 are chosen because the financial crisis occurred within the sample interval and we can observe the impact of these data on the model. After the financial crisis ended and the stock market gradually returned to normal, we select the last two start dates Aug 1 2008 and June 1 2009.} 

\subsection{Comparison of our model with conventional methods} 
In this section, we will compare the robust single and our 2-period models with other conventional methods like equal-weighted method, classical Markowitz model \citep{markowitz1952portfolio}, and the maxmin utility of mean-variance (UM) method \citep[Chapter 12]{fabozzi2007robust}.

Let $\mu_{\boldsymbol{R}}=E_P[\boldsymbol{R}]$, $\widebar{\mu_{\boldsymbol{R}}}=E_Q[\boldsymbol{R}]$, $K_{\boldsymbol{RR}}=E_P[\boldsymbol{RR}^\intercal]$, and $\overline{K_{\boldsymbol{RR}}}=E_Q[\boldsymbol{RR}^\intercal]$. The classical Markowitz mean-variance problem is to solve
\eqn{\max_{\boldsymbol{w}}\edg{ \boldsymbol{w}^\intercal \widebar{\mu_{\boldsymbol{R}}}-\frac{\gamma}{2} \boldsymbol{w}^\intercal \widebar{K_{\boldsymbol{RR}}} \boldsymbol{w}}\quad s.t\;\quad \textbf{1}^\intercal \boldsymbol{w}=1,
}
where $\gamma>0$ is a risk-aversion coefficient. One way of making the classical Markowitz model robust is by allowing the uncertainty in the mean return. In other words, for a uncertainty set $\scU$, we solve
\eqn{\max_{\boldsymbol{w}}\min_{\mu_{\boldsymbol{R}}\in \scU}\edg{ \boldsymbol{w}^\intercal \mu_{\boldsymbol{R}}-\frac{\gamma}{2} \boldsymbol{w}^\intercal \widebar{K_{\boldsymbol{RR}}} \boldsymbol{w}}\quad s.t\;\quad \textbf{1}^\intercal \boldsymbol{w}=1.
}
For example, \citet[Chapter 12]{fabozzi2007robust} considered the uncertainty set $$\scU=\scU_\delta(\widebar{\mu_{\boldsymbol{R}}})=\Big\{\mu_{\boldsymbol{R}}:(\mu_{\boldsymbol{R}}-\widebar{\mu_{\boldsymbol{R}}})^\intercal {K_{\mu_{\boldsymbol{R}}}}^{-1}(\mu_{\boldsymbol{R}}-\widebar{\mu_{\boldsymbol{R}}})\leq \delta^2 \Big\}.$$ Here, $K_{\mu_{\boldsymbol{R}}}$ is the covariance matrix of the errors in the estimation of the expected (average) returns. Under the assumption that the returns in a given sample of size $d$ are independent and identically distributed, we have $K_{\mu}=\widebar{K_{\boldsymbol{RR}}}/d$. The problem can be translated to its dual,
\begin{equation}\label{}
	\begin{aligned}
		&\max_{\boldsymbol{w}}\;\boldsymbol{w}^\intercal \widebar{\mu_{\boldsymbol{R}}}-\frac{\gamma}{2} \boldsymbol{w}^\intercal \widebar{K_{\boldsymbol{RR}}} \boldsymbol{w}-\delta\sqrt{\boldsymbol{w}^\intercal K_{\mu_{\boldsymbol{R}}} \boldsymbol{w}}\\
		&\quad s.t\;\;\textbf{1}^\intercal \boldsymbol{w}=1.
	\end{aligned}
\end{equation}
The $\delta$ could be specified by assuming some confidence interval around the historical expected return. Let $\sigma_d$ be the sample's standard deviation. Then, we let the $\delta=1.96$, which corresponds to the $95\%$ confidence interval for $\mu_{\boldsymbol{R}}$ by the Central Limit Theorem. We set $d=500$ (previous 2 year) and $\gamma=4$ for our simulations.

\subsubsection*{Simulation Results}
Table \ref{table2} shows the result of the optimal investment without transaction costs for the five methods.
\begin{table}[H]\centering
	\small
	\begin{tabular}{ |c||c|c|c|  }
		\hline
		2002.02.01 & Mean (Daily) & Std Dev (Daily) & Sharpe (Annualized)\\
		\hline
		equal-weighted & 0.000560395 & 0.014345112 & 0.620141254\\
		Markowitz & 0.000413440 & 0.010377771 & 0.632424947\\
		UM & 0.000426882 & 0.010438333 & 0.649199194\\
		single-period & 0.000545715 & 0.012848793 & 0.674223447
		\\
		2-period & 0.000385135 & 0.007563152 & 0.808371140\\
		\hline
		2004.06.01 & Mean (Daily) & Std Dev (Daily) & Sharpe (Annualized)\\
		\hline
		equal-weighted & 0.000508735 & 0.014763839 & 0.547007214\\
		Markowitz & 0.000472901 & 0.010955220 & 0.685251078\\
		UM & 0.000482110 & 0.011082446 & 0.690576211\\
		single-period & 0.000423354 & 0.011642227 & 0.577255661\\
		2-period & 0.000447016 & 0.007866766 & 0.902044465\\
		\hline
		2006.06.01 & Mean (Daily) & Std Dev (Daily) & Sharpe (Annualized)\\
		\hline
		equal-weighted & 0.000566047 & 0.014870618 & 0.604260345\\
		Markowitz & 0.000595641 & 0.009098311 & 1.039260405\\
		UM & 0.000609014 & 0.009163261 & 1.055061334\\
        single-period & 0.000521979 & 0.009156240 & 0.904975460\\
		2-period & 0.000516936 & 0.007580369 & 1.082548396\\
		\hline
		2008.08.01 & Mean (Daily) & Std Dev (Daily) & Sharpe (Annualized)\\
		\hline
		equal-weighted & 0.000582780 & 0.014381476 & 0.643281588\\
		Markowitz & 0.000639435 & 0.007996818 & 1.269344515\\
		UM & 0.000647746 & 0.008006844 & 1.284233376\\
		single-period & 0.000624190 & 0.007845931 & 1.262911166\\
		2-period & 0.000629368 & 0.008288125 & 1.205447689\\
		\hline
		2009.06.01	& Mean (Daily) & Std Dev (Daily) & Sharpe (Annualized)\\
		\hline
		equal-weighted & 0.000683486 & 0.009834511 & 1.103257959\\
		Markowitz & 0.000517496 & 0.006582772 & 1.247953883\\
		UM & 0.000531380 & 0.006568064 & 1.284305944\\
		single-period & 0.000637698 & 0.007654071 & 1.322583602\\
		2-period & 0.000555807 & 0.006527001 & 1.351793232\\
		\hline
	\end{tabular}
	\caption{100 Stocks without transaction costs}
	\label{table2}
\end{table}
The following plots are the rolling 1-year Sharpe ratios of each strategy without transaction costs, starting from different dates.

\begin{figure}[H]
	\centering
	\includegraphics[width=14cm]{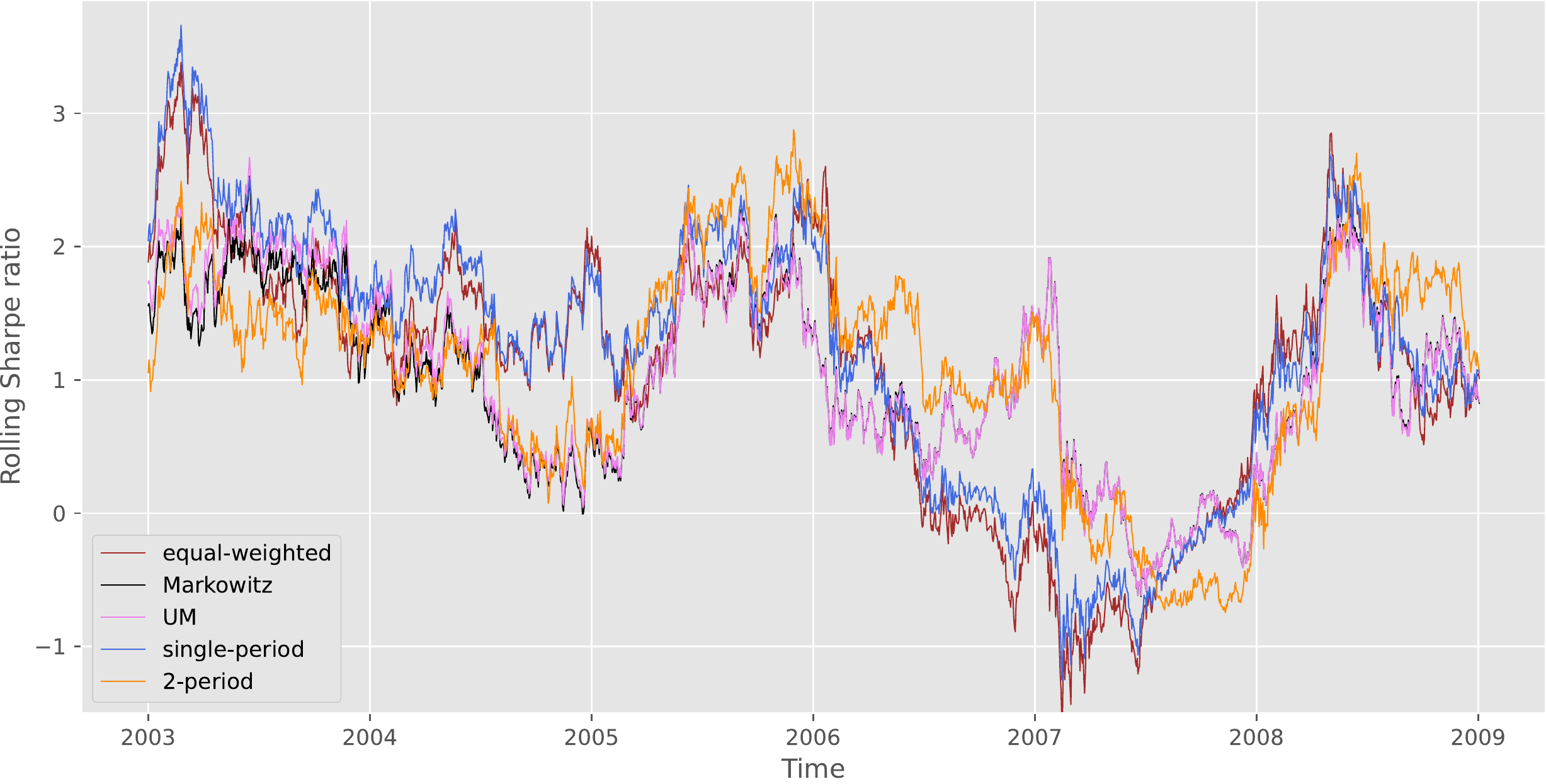}
	\caption{100 Stocks without transaction costs starting 2002.02.01}
	\label{fig1}
\end{figure}
\begin{figure}[H]
	\centering
	\includegraphics[width=14cm]{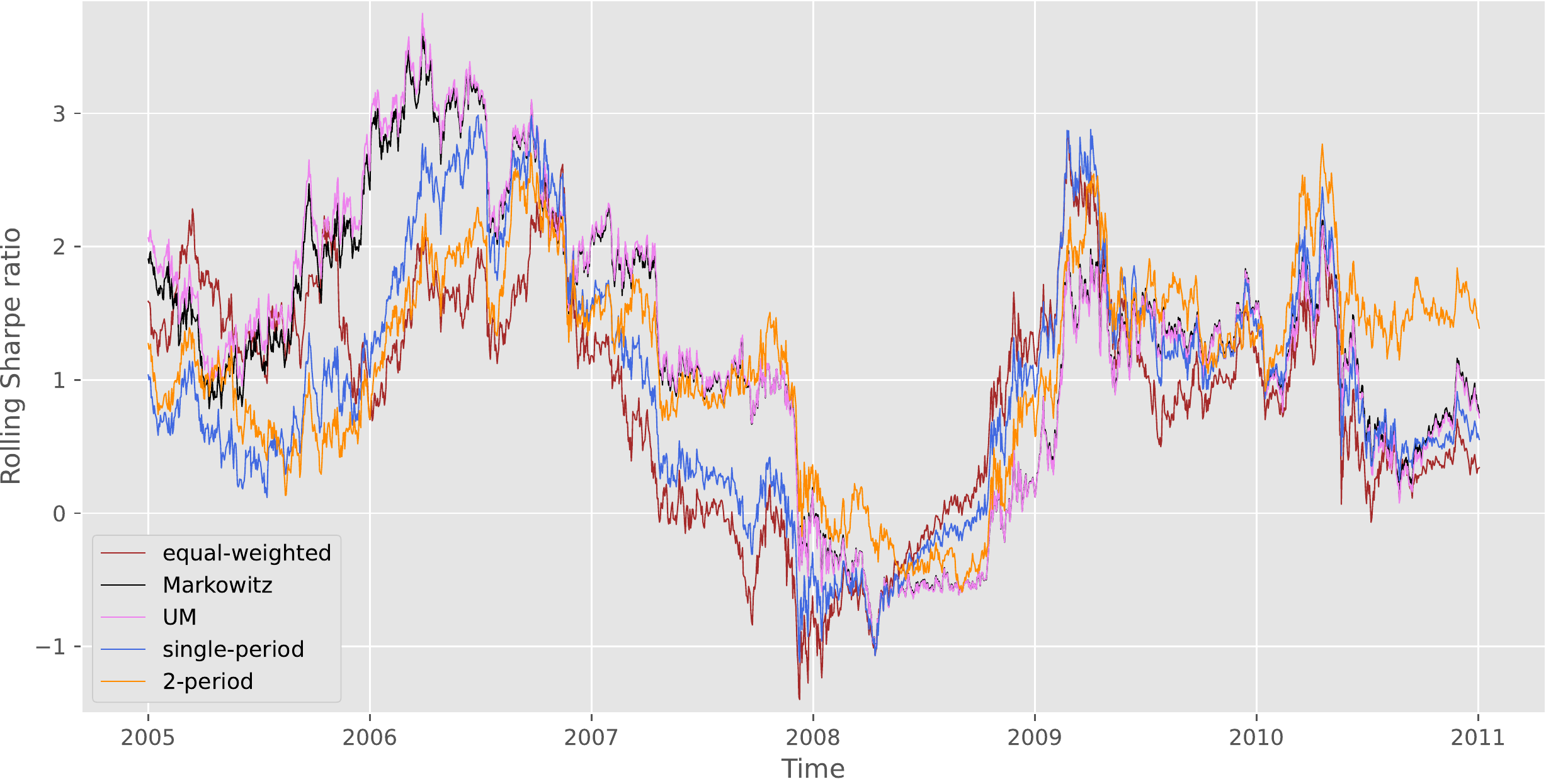}
	\caption{100 Stocks without transaction costs starting 2004.06.01}
	\label{fig1.0.1}
\end{figure}
\begin{figure}[H]
	\centering
	\includegraphics[width=14cm]{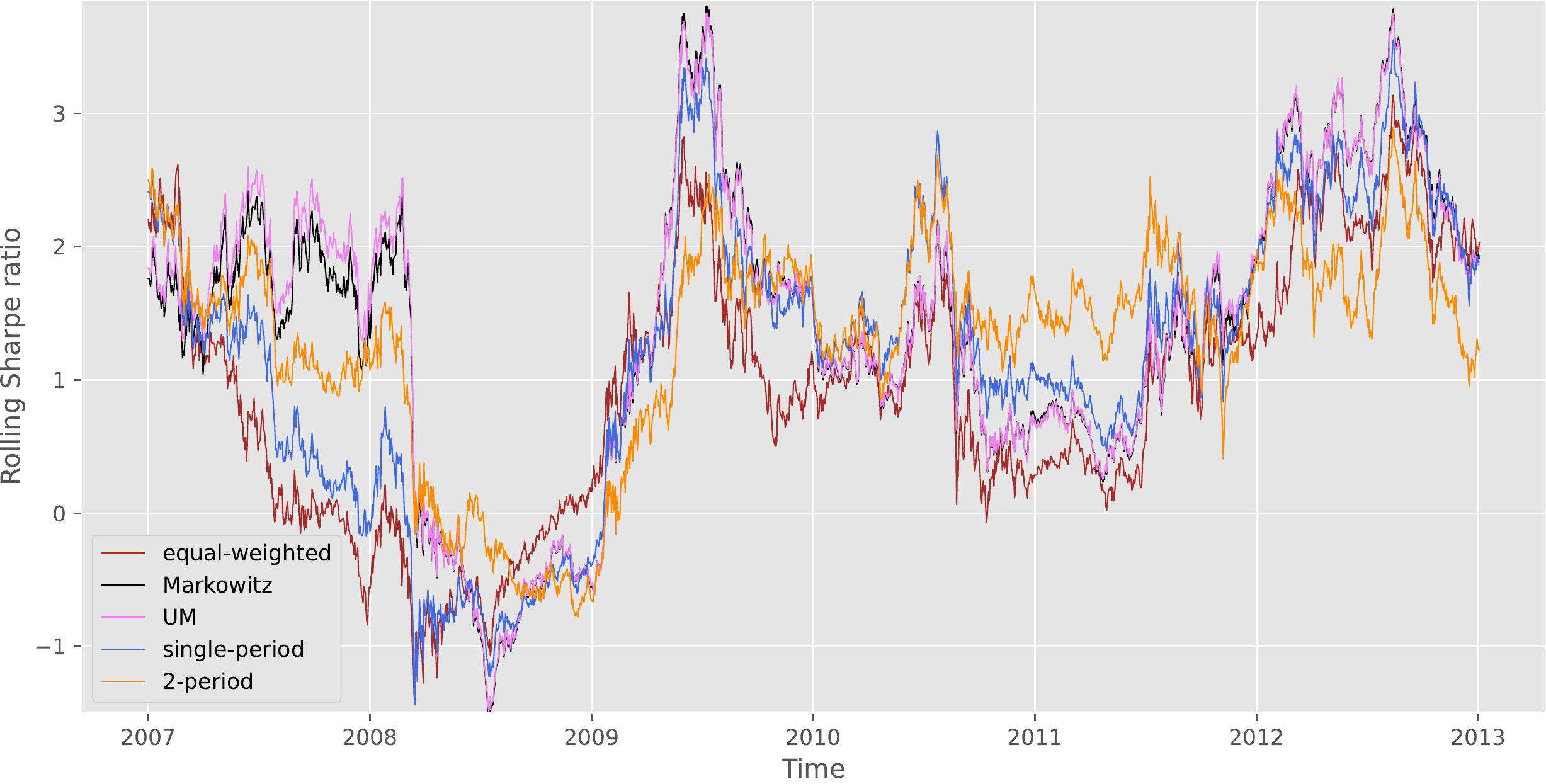}
	\caption{100 Stocks without transaction costs starting 2006.06.01}
	\label{fig1.0.2}
\end{figure}
\begin{figure}[H]
	\centering
	\includegraphics[width=14cm]{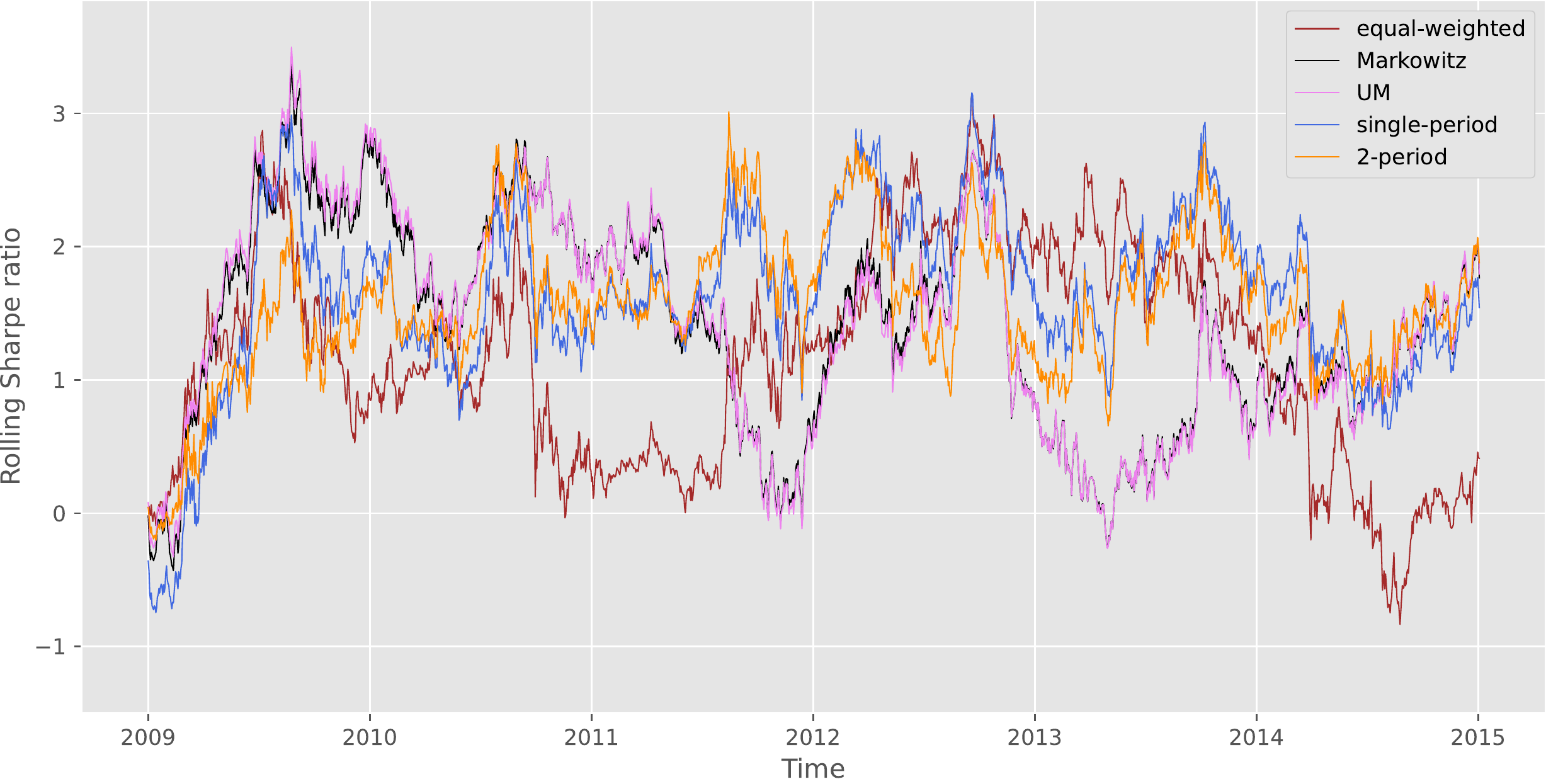}
	\caption{100 Stocks without transaction costs starting 2008.08.01}
	\label{fig2}
\end{figure}
\begin{figure}[H]
	\centering
	\includegraphics[width=14cm]{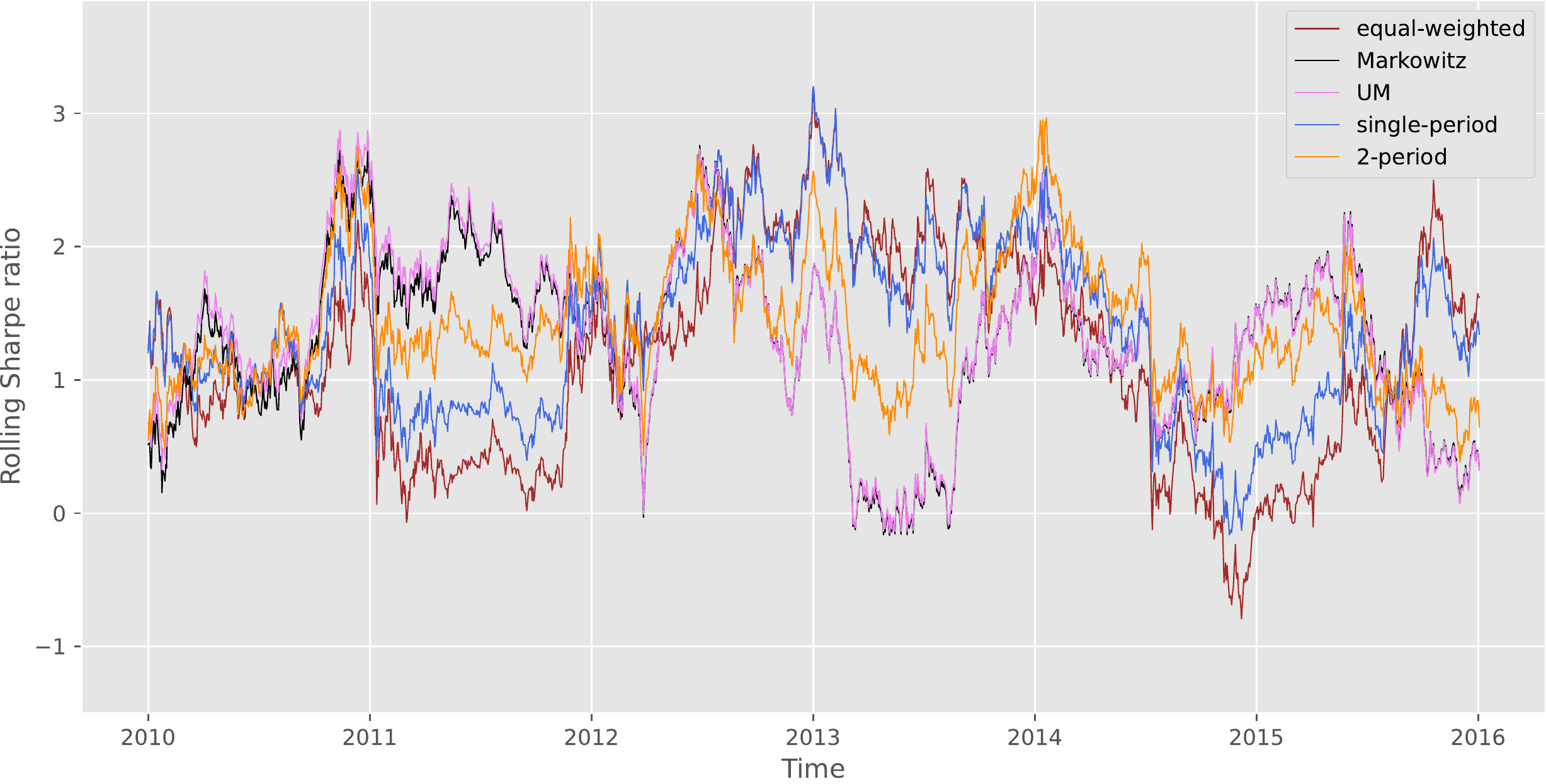}
	\caption{100 Stocks without transaction costs starting 2009.06.01}
	\label{fig3}
\end{figure}
Table \ref{table3} is the comparison between the investment strategies with the transaction costs.
\begin{table}[H]\centering
	\small
	\begin{tabular}{ |c||c|c|c|  }
		\hline
		2002.02.01 & Mean (Daily) & Std Dev (Daily) & Sharpe (Annualized)\\
		\hline
		equal-weighted & 0.000545436 & 0.014333250 & 0.604086435\\
		Markowitz & 0.000398957 & 0.010382485 & 0.609992667\\
		UM & 0.000412398 & 0.010442989 & 0.626891504\\
		single-period & 0.000531535 & 0.012841564 & 0.657074519\\
		2-period & 0.000377872 & 0.007573274 & 0.792065680\\
		\hline
		2004.06.01 & Mean (Daily) & Std Dev (Daily) & Sharpe (Annualized)\\
		\hline
		equal-weighted & 0.000504211 & 0.014754810 & 0.542474635\\
		Markowitz & 0.000463776 & 0.010955203 & 0.672030154\\
		UM & 0.000472937 & 0.011082582 & 0.677427621\\
		single-period & 0.000415804 & 0.011634557 & 0.567335551\\
		2-period & 0.000449246 & 0.007868584 & 0.906335002\\
		\hline
		2006.06.01 & Mean (Daily) & Std Dev (Daily) & Sharpe (Annualized)\\
		\hline
		equal-weighted & 0.000562507 & 0.014862114 & 0.600824808\\
		Markowitz & 0.000584489 & 0.009120416 & 1.017330936\\
		UM & 0.000598274 & 0.009186519 & 1.033832131\\
		single-period & 0.000514421 & 0.009149157 & 0.892561938\\
		2-period & 0.000512448 & 0.007584638 & 1.072544490\\
		\hline
		2008.08.01 & Mean (Daily) & Std Dev (Daily) & Sharpe (Annualized)\\
		\hline
		equal-weighted & 0.000582874 & 0.014372741 & 0.643777230\\
		Markowitz & 0.000637931 & 0.008029348 & 1.261229174\\
		UM & 0.000645801 & 0.008039414 & 1.275190017\\
		single-period & 0.000622847	& 0.007849569 & 1.259609801\\
		2-period & 0.000628367 & 0.008280399 & 1.204653856
		\\
		\hline
		2009.06.01	& Mean (Daily) & Std Dev (Daily) & Sharpe (Annualized)\\
		\hline
		equal-weighted & 0.000667442 & 0.009817072 & 1.079273577\\
		Markowitz & 0.000515463 & 0.006605026 & 1.238862087\\
		UM & 0.000529484 & 0.006590397 & 1.275386958\\
		single-period & 0.000625088 & 0.007648792 & 1.297326189\\
		2-period & 0.000546146 & 0.006535765 & 1.326516534
		\\
		\hline
	\end{tabular}
	\caption{100 Stocks with transaction costs}
	\label{table3}
\end{table}

The following plots are the rolling 1-year Sharpe ratios of each strategy with transaction costs, starting from different dates.
\begin{figure}[H]
	\centering
	\includegraphics[width=14cm]{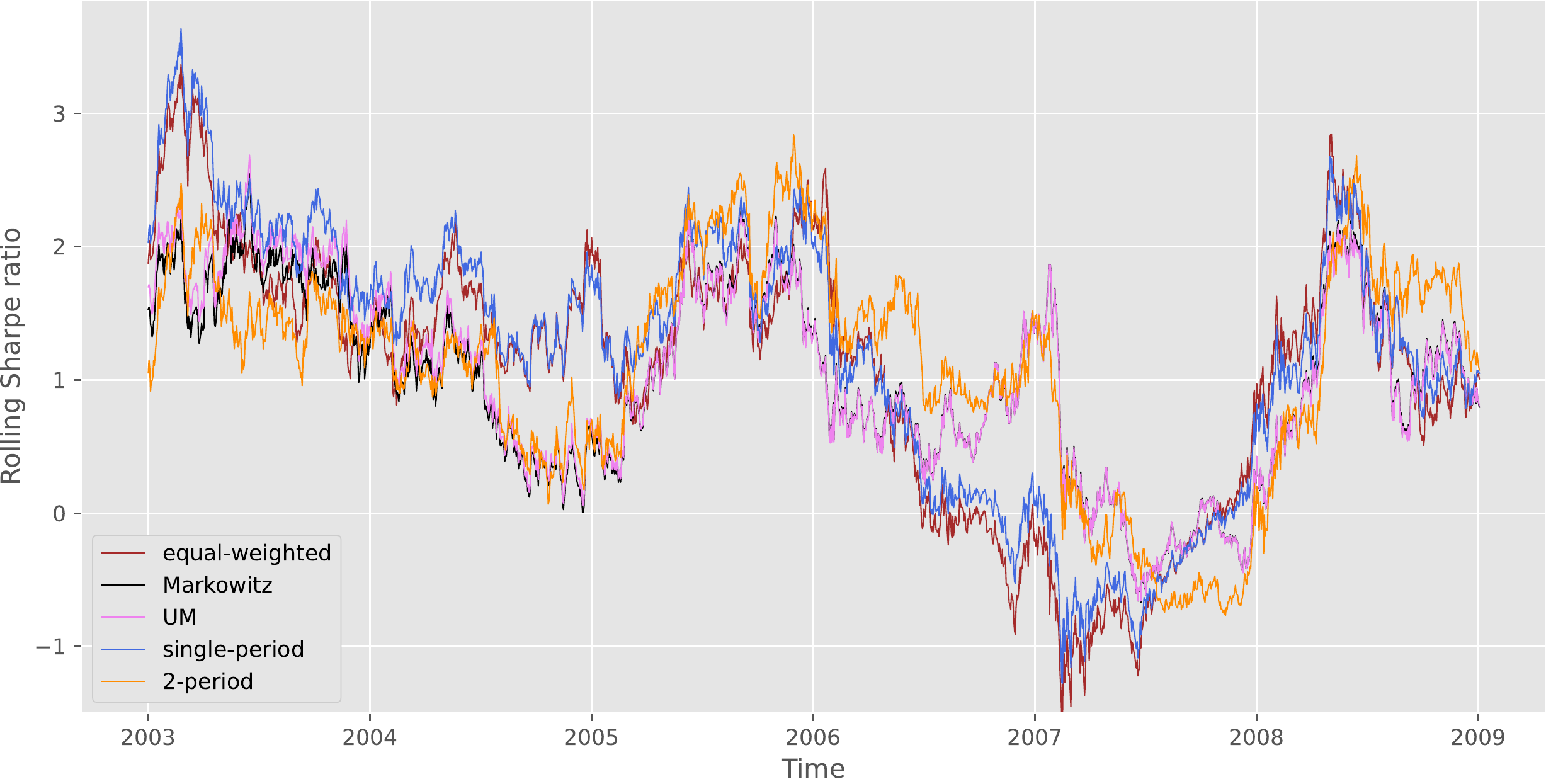}
	\caption{100 Stocks with transaction costs starting 2002.02.01}
	\label{fig4}
\end{figure}
\begin{figure}[H]
	\centering
	\includegraphics[width=14cm]{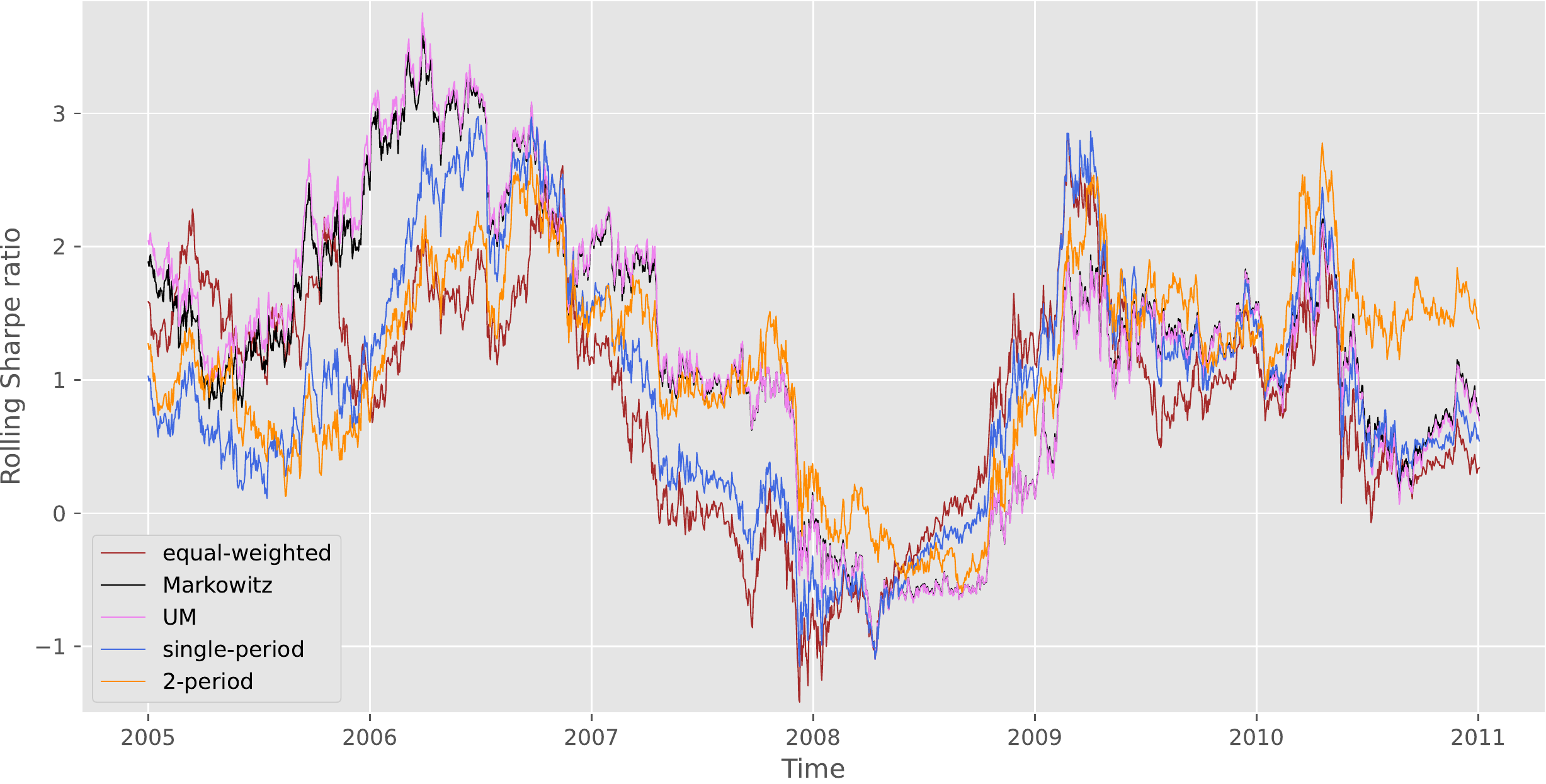}
	\caption{100 Stocks with transaction costs starting 2004.06.01}
	\label{fig1.0.3}
\end{figure}
\begin{figure}[H]
	\centering
	\includegraphics[width=14cm]{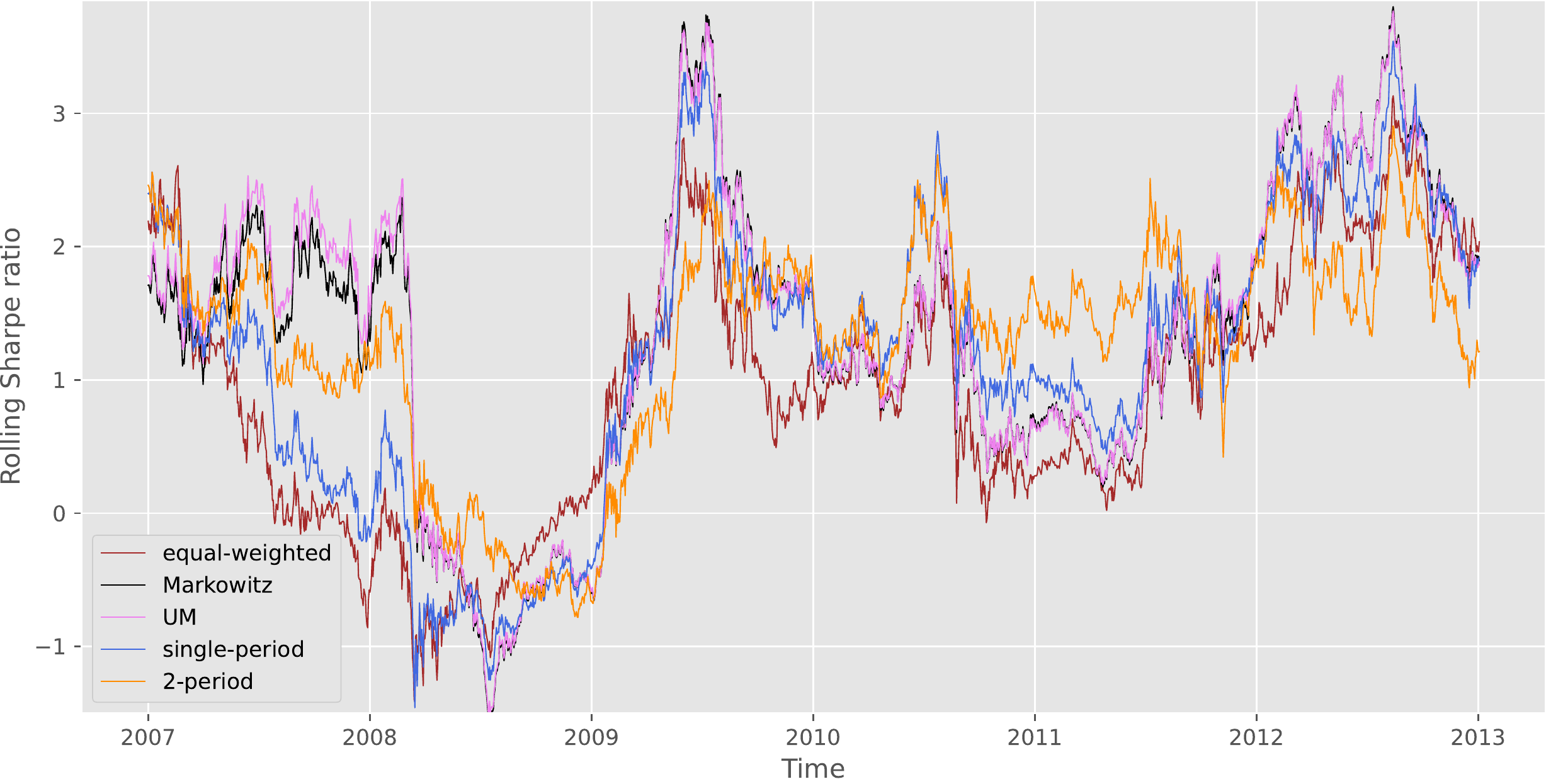}
	\caption{100 Stocks with transaction costs starting 2006.06.01}
	\label{fig1.0.4}
\end{figure}

\begin{figure}[H]
	\centering
	\includegraphics[width=14cm]{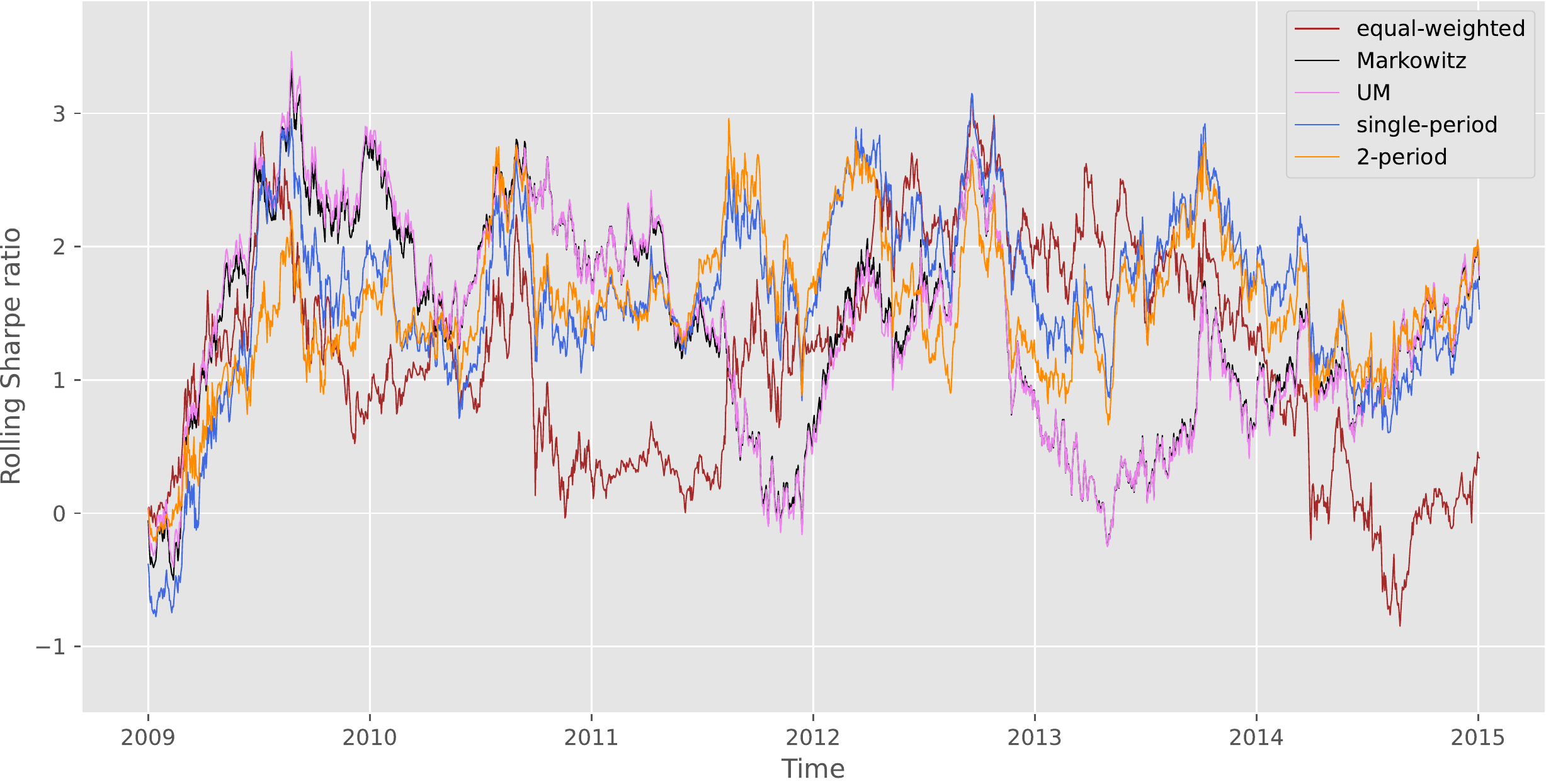}
	\caption{100 Stocks with transaction costs starting 2008.08.01}
	\label{fig5}
\end{figure}
\begin{figure}[H]
	\centering
	\includegraphics[width=14cm]{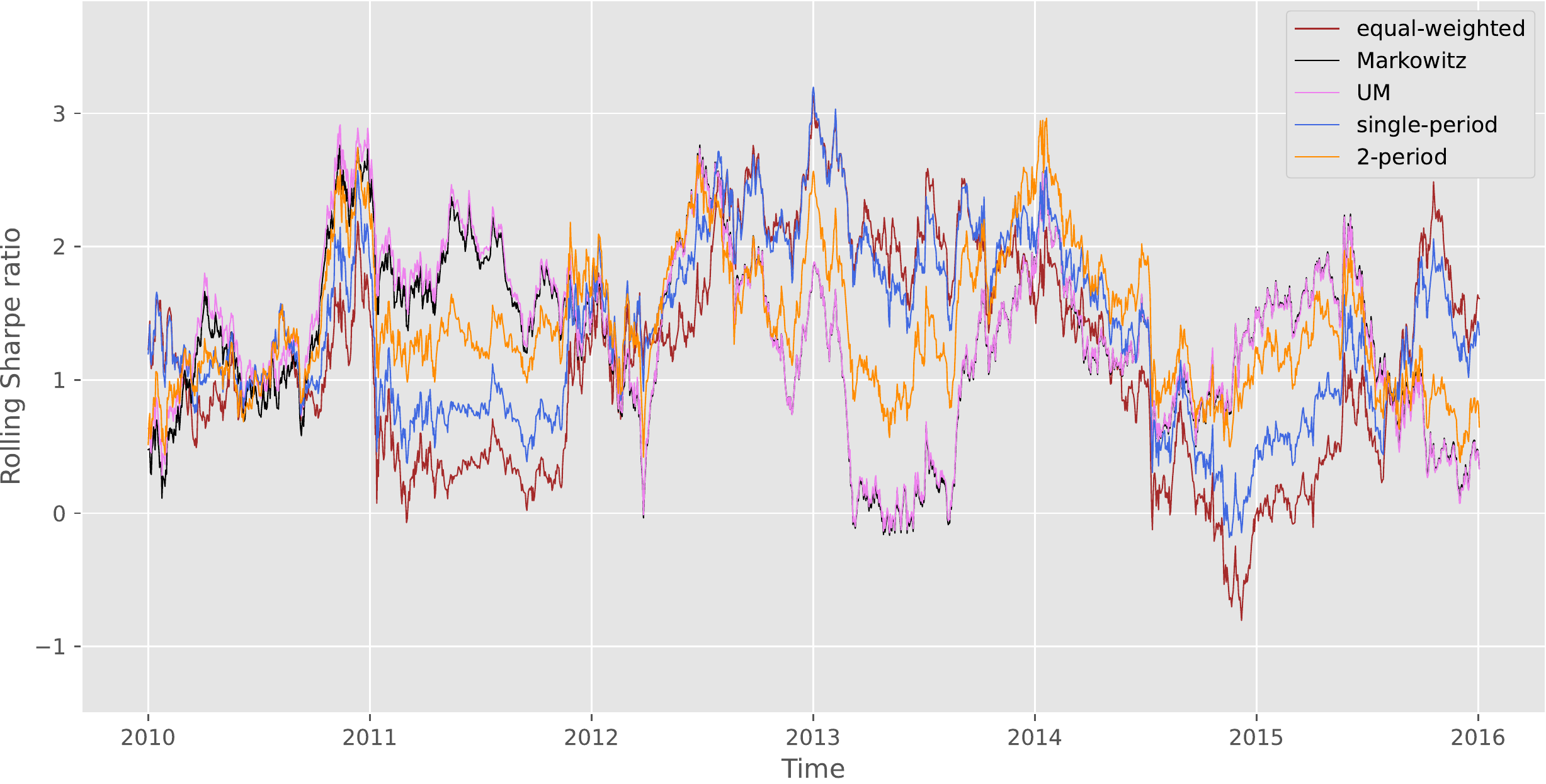}
	\caption{100 Stocks with transaction costs starting 2009.06.01}
	\label{fig6}
\end{figure}
\subsection{Comparison of 2-period model and robust minimum variance methods} 
In this subsection, we use the daily adjusted closing price of the 100 largest S\&P 500 companies by market cap and compare our 2-period model with different robust minimum variance single-period methods ($T=1$), which are Blanchet's single-period model (\cite{blanchet2021distributionally}), robust norm-constrained (NC) model (\cite{xing2014robust}), and sparse and robust mean-variance (SP) model (\cite{dai2019sparse}). The rebalancing of each strategy follows the same algorithm presented in Section \ref{section5}.

Let us briefly summarize NC and SP models. We let $\boldsymbol{w}=(w_1,...,w_n)$ be the vector of portfolio weights, $\widebar{K_{\boldsymbol{RR}}}$ be
the sample covariance matrix which is the estimation of the asset return covariance matrix $K_{\boldsymbol{RR}}$. \cite{xing2014robust} used the NC model as a robust method to control the sparsity or shrink of the estimated weights of assets. More precisely, they studied 
\begin{equation}
	\begin{aligned}
		&\min_{\boldsymbol{w}}\quad \boldsymbol{w}^\intercal {\widebar{K_{\boldsymbol{RR}}}} \boldsymbol{w}\\
		&\;s.t\quad \bf{1}^\intercal \boldsymbol{w}=1\text{ and } \norm{ \boldsymbol{w}}_1+\alpha\norm{\boldsymbol{w}}_\infty\leq c
	\end{aligned}
\end{equation}
where  $\alpha\geq 0 $ is the penalty parameter and $c$ is a constraint. They proved that the problem is equivalent to
\begin{equation}\label{Xing2}
\begin{aligned}
	&\min_{\boldsymbol{w}}\quad \boldsymbol{w}^\intercal \widehat{K_{\boldsymbol{RR}}}\boldsymbol{w}\\
	&\;s.t\quad {\bf 1}^\intercal \boldsymbol{w}=1\text{ and } \widehat{K_{\boldsymbol{RR}}} = \widebar{K_{\boldsymbol{RR}}} - \lambda(\boldsymbol{v1}^\intercal+\boldsymbol{1v}^\intercal)+\frac{1}{2}\lambda\alpha(\boldsymbol{o1}^\intercal+\boldsymbol{1o}^\intercal),
\end{aligned}
\end{equation}
where $\lambda$ is the penalty parameters determined by $c$, $\boldsymbol{v}$ is an $n\times 1$ vector whose $i$th component is 1 if $w_i$ is negative and is 0 otherwise, and $\boldsymbol{v}$ is an $n\times 1$ vector whose $i$th component is ${\rm sign}(w_i)$ if $w_i$ has the largest absolute value, and is $0$ otherwise. The vector $\boldsymbol{o}$ is an $n\times 1$ vector with the $i$-$th$ component being $o_i=({\rm rank}(|w_i|)-1){\rm sign}(w_i)$, and ${\rm rank}(|w_i|)$ is the rank of $|w_i|$ in $\crl{|w_1|,...,|w_n|}$. According to their simulation setup, we select risk aversion $\lambda=2$ and tuning parameter $\alpha=4$.

SP model aimed to reduce the undesired impact of parameter uncertainty and estimation errors of the mean-variance portfolio model. They proposed a new sparse and robust portfolio optimization model 
\begin{equation}
	\begin{aligned}
		&\min_{\boldsymbol{w}}\quad \boldsymbol{w}^\intercal \overline{K_{\boldsymbol{RR}}}\boldsymbol{w} + \frac{\tau}{T}\|\boldsymbol{w}\|_1\\
		&\;s.t\quad \min_{E_P[\boldsymbol{R}]\in \scM}\boldsymbol{w}^\intercal E_P[\boldsymbol{R}] = \rho\text{ and } \boldsymbol{1}^\intercal \boldsymbol{w}=1,
	\end{aligned}
\end{equation}
where $\scM$ denotes the uncertainty set for $E_P[\boldsymbol{R}]$ and $\tau$ is a parameter that allows investors to adjust the relative importance of the $l_1$ penalization in the optimization. For simplicity, let us consider the box uncertainty set $\scM_B=\{\mu:\mu=E_Q[\boldsymbol{R}]+\xi,|\xi_i|\leq\epsilon_i,i=1,...n\}$. Then, it is easy to verify the formulation is equivalent to
\begin{equation}
	\begin{aligned}
		&\min_{\boldsymbol{w}}\quad \|\rho \boldsymbol{1}-R_{d\times n}\boldsymbol{w}\|^2_2 + \tau\|\boldsymbol{w}\|_1\\
		&\;s.t\quad \boldsymbol{w}^\intercal E_Q[\boldsymbol{R}]-\epsilon^\intercal |\boldsymbol{w}| = \rho \text{ and } \boldsymbol{1}^\intercal \boldsymbol{w}=1,
	\end{aligned}
\end{equation}
where $d$ is sample size, $R_{d\times n}$ is $d\times n$ matrix of which row $t$ equals $\boldsymbol{R}_t$, and $|\boldsymbol{w}|=(|w_1|,...,|w_n|)^\intercal$.
They choose $\tau=120$ for monthly data corresponds to 10 years. 
In our simulation, we set $\tau=500$ since we are using the daily data for 2 years to find $\boldsymbol{w}$.

\subsubsection*{Simulation Results}
Table \ref{table4} is the comparison between the investment strategies without the transaction costs.
\begin{table}[H]\centering
\small
\begin{tabular}{ |c||c|c|c|  }
 \hline
2002.02.01 & Mean (Daily) & Std Dev (Daily) & Sharpe (Annualized)\\
 \hline
NC & 0.000560395 & 0.014345112 & 0.620141255\\
SP & 0.000548045 & 0.013289725 & 0.654637011 \\
single-period & 0.000545715 & 0.012848793 & 0.674223447
\\
2-period & 0.000385135 & 0.007563152 & 0.808371140\\
\hline
2004.06.01 & Mean (Daily) & Std Dev (Daily) & Sharpe (Annualized)\\
\hline
NC & 0.000508735 & 0.014763839 & 0.547007214\\
SP & 0.000509687 & 0.013512330 & 0.598789645 \\
single-period & 0.000423354 & 0.011642227 & 0.577255661\\
2-period & 0.000447016 & 0.007866766 & 0.902044465\\
\hline
2006.06.01 & Mean (Daily) & Std Dev (Daily) & Sharpe (Annualized)\\
\hline
NC & 0.000566047 & 0.014870618 & 0.604260345\\
SP & 0.000564831 & 0.012963183 & 0.691684099 \\
single-period & 0.000521979 & 0.009156240 & 0.904975460\\
2-period & 0.000516936 & 0.007580369 & 1.082548396\\
 \hline
 2008.08.01 & Mean (Daily) & Std Dev (Daily) & Sharpe (Annualized)\\
 \hline
NC & 0.000582780 & 0.014381476 & 0.643281589\\
SP & 0.000606592 & 0.012224146 & 0.787731381\\
single-period & 0.000624190 & 0.007845931 & 1.262911166\\
2-period & 0.000629368 & 0.008288125 & 1.205447689\\
\hline
2009.06.01	& Mean (Daily) & Std Dev (Daily) & Sharpe (Annualized)\\
\hline
NC & 0.000683486 & 0.009834511 & 1.103257959\\
SP & 0.000656613 & 0.008177180 & 1.274694655\\
single-period & 0.000637698 & 0.007654071 & 1.322583602\\
2-period & 0.000555807 & 0.006527001 & 1.351793232\\
\hline
\end{tabular}
\caption{100 Stocks without transaction costs}
\label{table4}
\end{table}
The following plots are the rolling 1-year Sharpe ratios of each strategy without transaction costs, starting from different dates.
\begin{figure}[H]
\centering
\includegraphics[width=14cm]{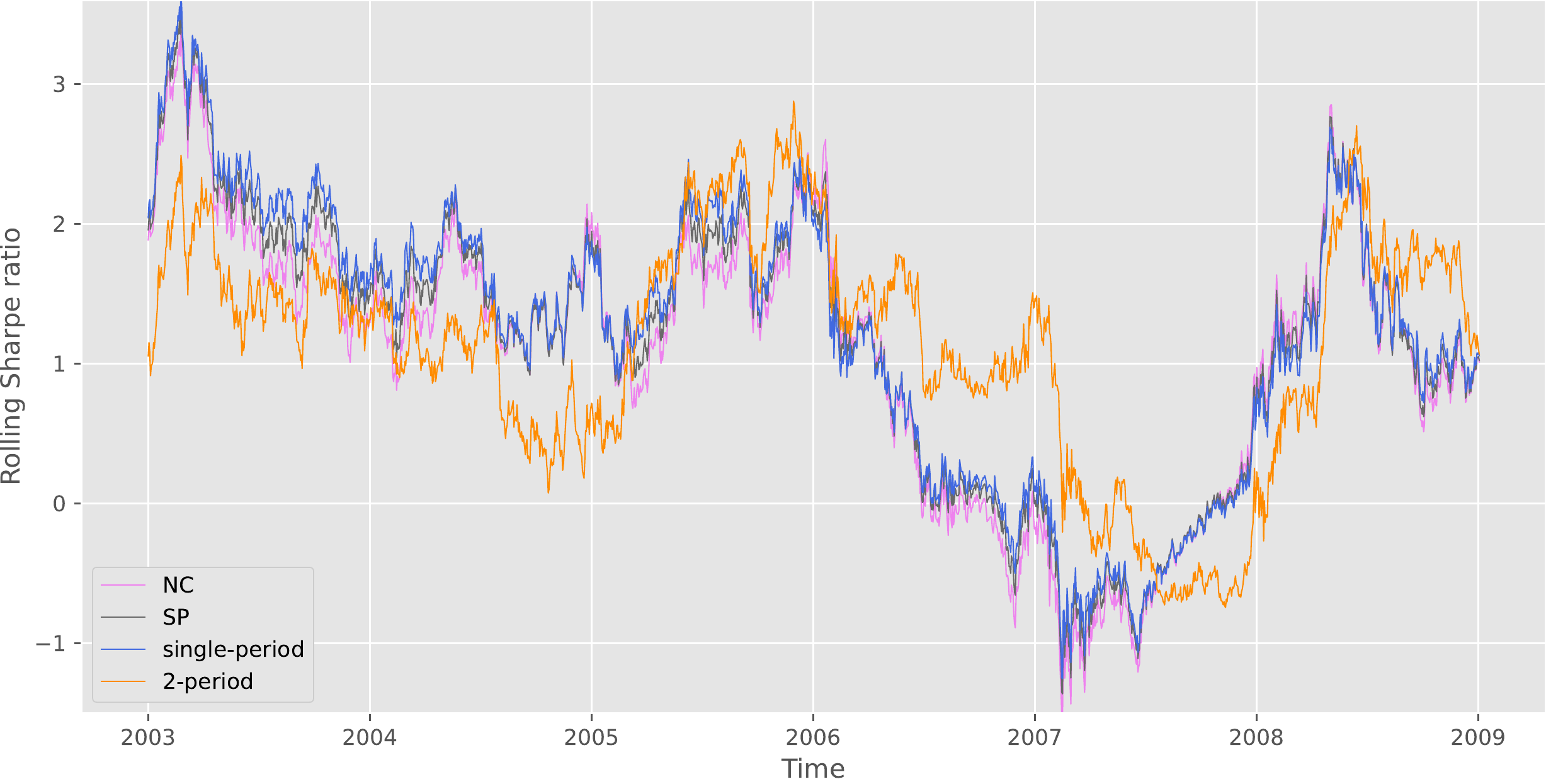}
\caption{100 Stocks without transaction costs, starting 2002.02.01}
\label{fig7}
\end{figure}
\begin{figure}[H]
\centering
\includegraphics[width=14cm]{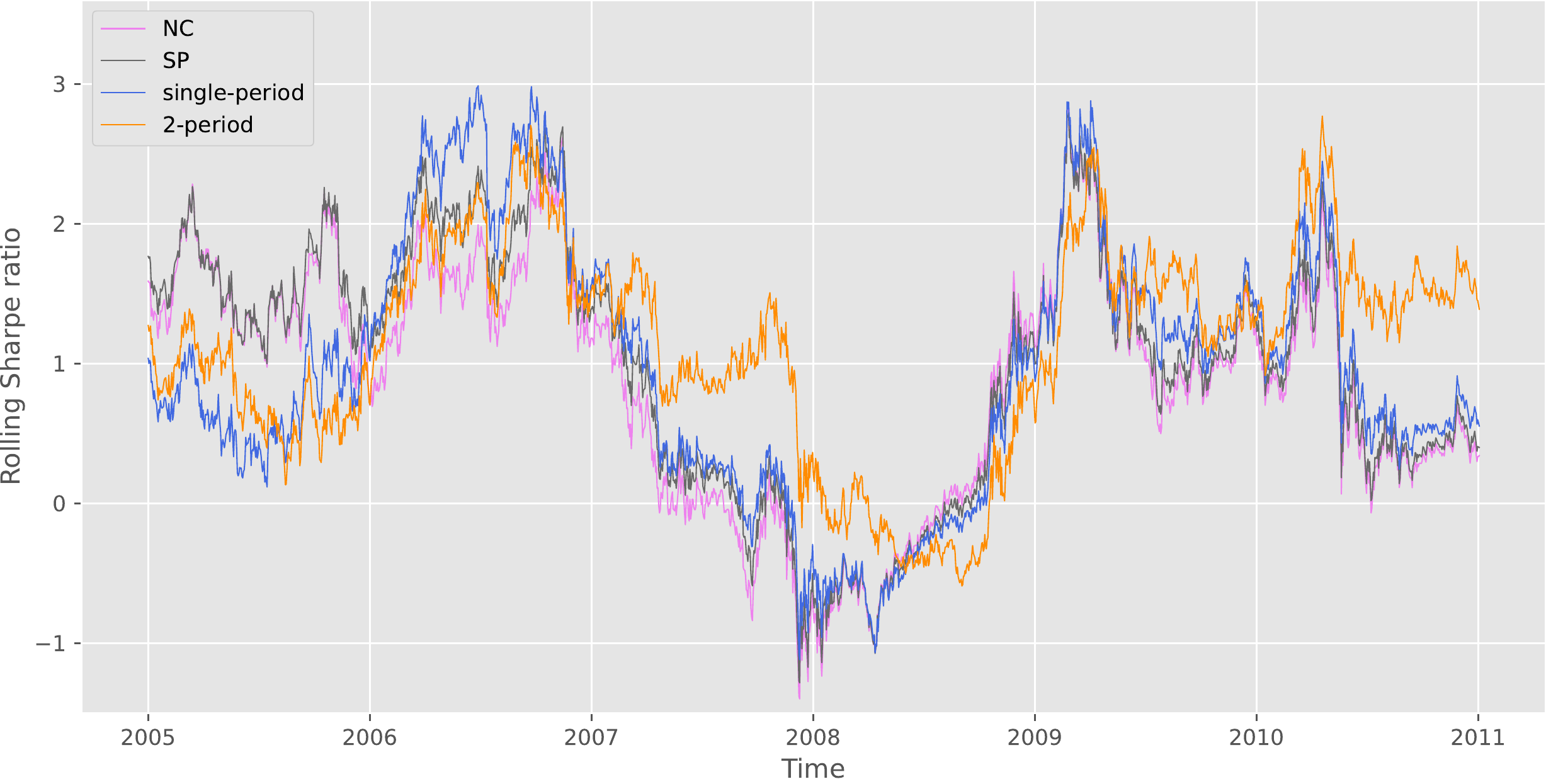}
\caption{100 Stocks without transaction costs, starting 2004.06.01}
\label{fig1.0.5}
\end{figure}
\begin{figure}[H]
\centering
\includegraphics[width=14cm]{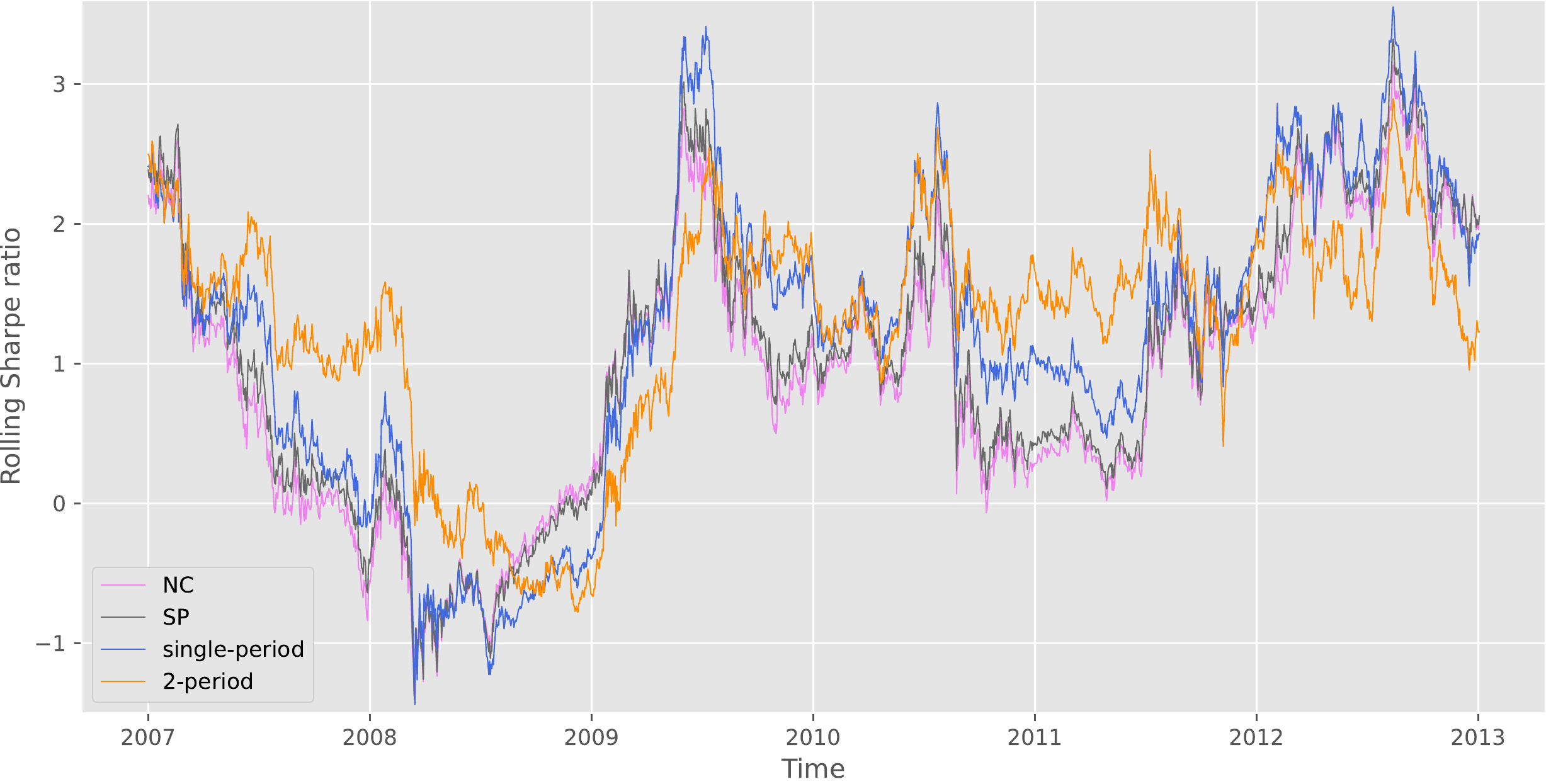}
\caption{100 Stocks without transaction costs, starting 2006.06.01}
\label{fig1.0.6}
\end{figure}

\begin{figure}[H]
\centering
\includegraphics[width=14cm]{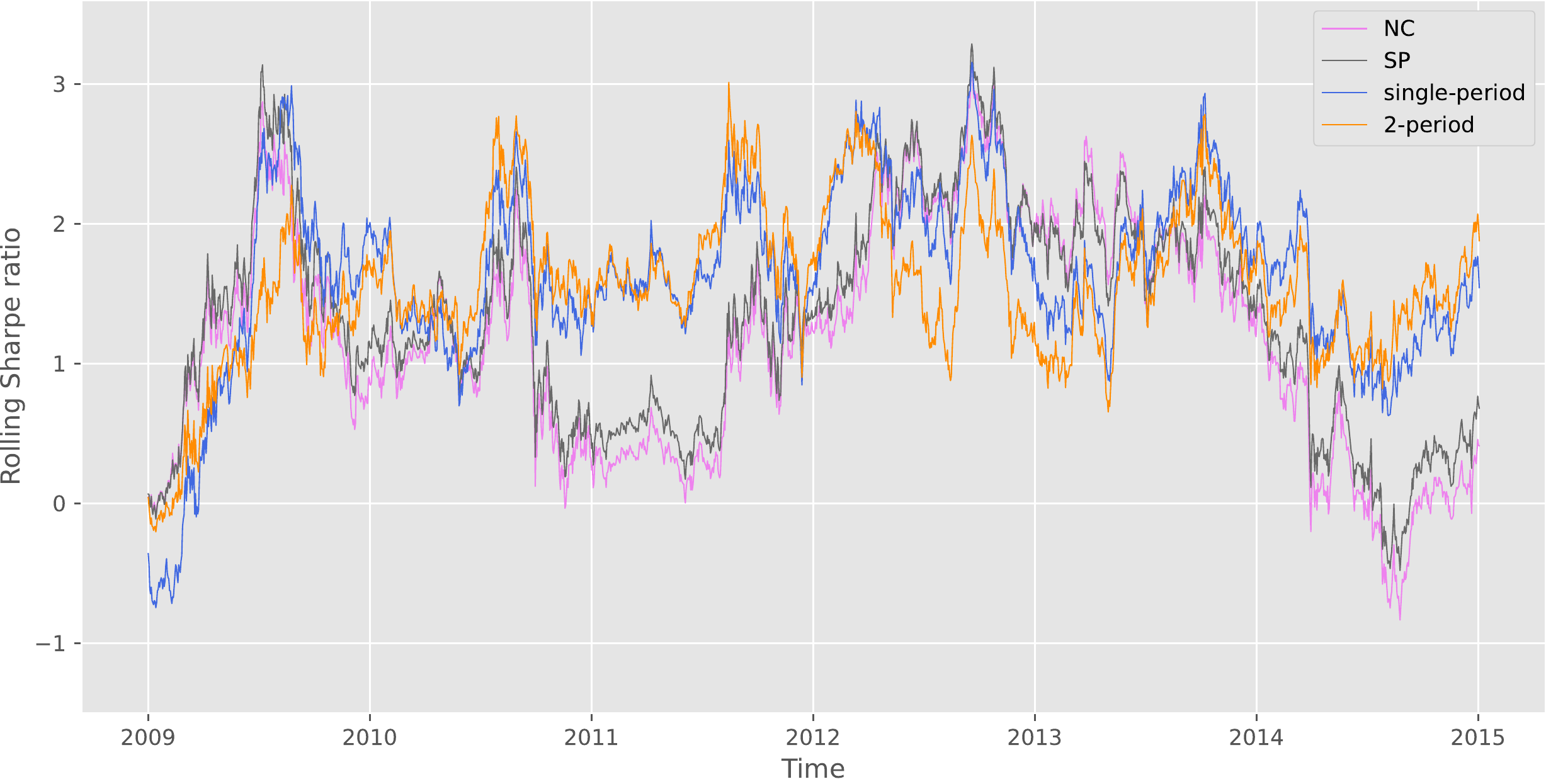}
\caption{100 Stocks without transaction costs, starting 2008.08.01}
\label{fig8}
\end{figure}
\begin{figure}[H]
\centering
\includegraphics[width=14cm]{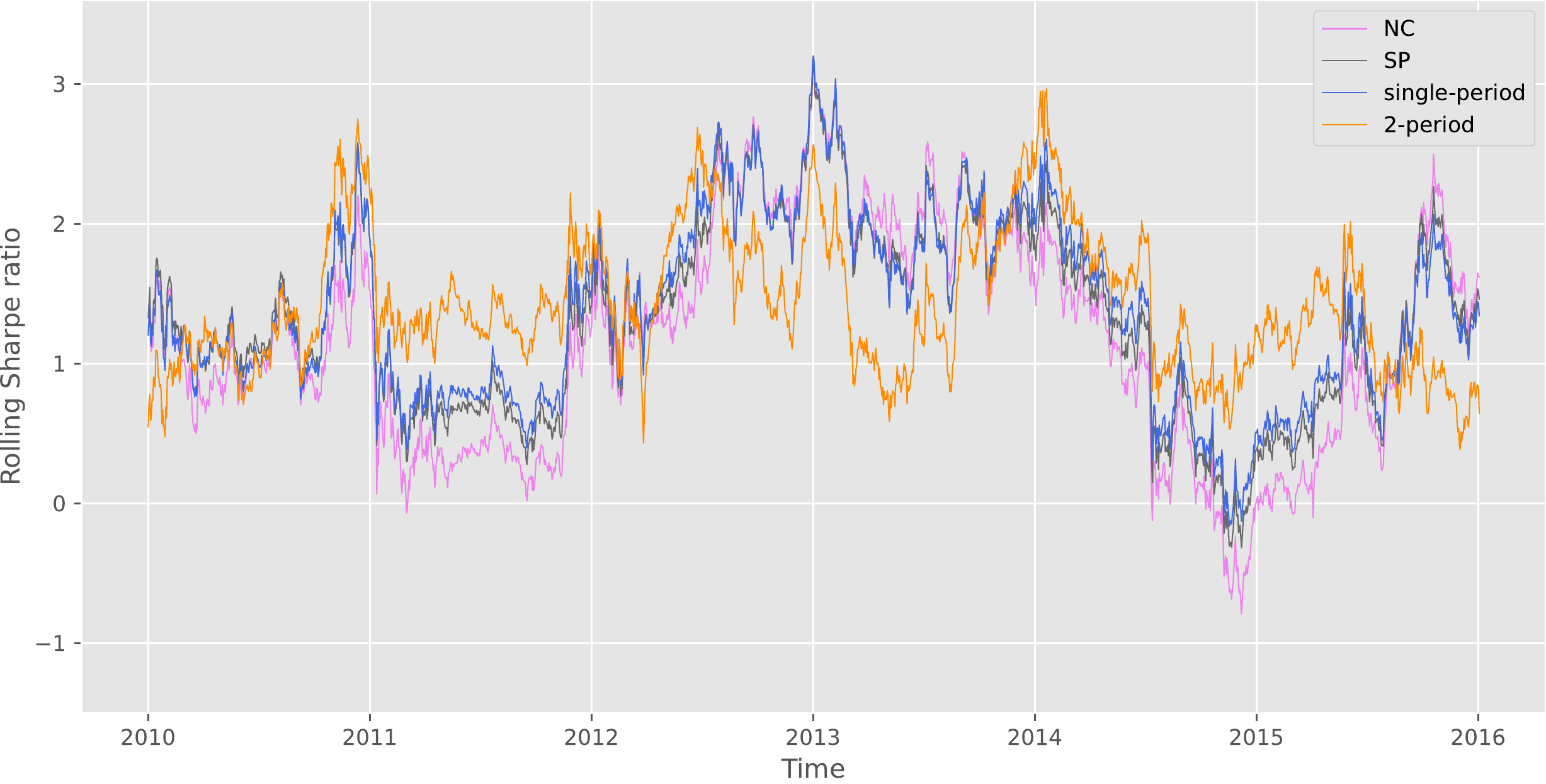}
\caption{100 Stocks without transaction costs, starting 2009.06.01}
\label{fig9}
\end{figure}

Table \ref{table5} is the comparison between the investment strategies with the transaction costs.
\begin{table}[H]\centering
\small
\begin{tabular}{ |c||c|c|c|  }
 \hline
2002.02.01 & Mean (Daily) & Std Dev (Daily) & Sharpe (Annualized)\\
 \hline
NC & 0.000545436 & 0.014333250 & 0.604086436\\
SP & 0.000533860 & 0.013281575 & 0.638084474\\
single-period & 0.000531535 & 0.012841564 & 0.657074519\\
2-period & 0.000377872 & 0.007573274 & 0.792065680\\
\hline
2004.06.01 & Mean (Daily) & Std Dev (Daily) & Sharpe (Annualized)\\
\hline
NC & 0.000504211 & 0.014754810 & 0.542474635\\
SP & 0.000504155 & 0.013504141 & 0.592648936 \\
single-period & 0.000415804 & 0.011634557 & 0.567335551\\
2-period & 0.000449246 & 0.007868584 & 0.906335002\\
\hline
2006.06.01 & Mean (Daily) & Std Dev (Daily) & Sharpe (Annualized)\\
\hline
NC & 0.000562506 & 0.014862119 & 0.600824228\\
SP & 0.000559966 & 0.012954736 & 0.686173440 \\
single-period & 0.000514421 & 0.009149157 & 0.892561938\\
2-period & 0.000512448 & 0.007584638 & 1.072544490\\
 \hline
 2008.08.01 & Mean (Daily) & Std Dev (Daily) & Sharpe (Annualized)\\
 \hline
NC & 0.000582874 & 0.014372741 & 0.643777231\\
SP & 0.000605829 & 0.012215743 & 0.787281677\\
single-period & 0.000622847	& 0.007849569 & 1.259609801\\
2-period & 0.000628367 & 0.008280399 & 1.204653856\\
\hline
2009.06.01	& Mean (Daily) & Std Dev (Daily) & Sharpe (Annualized)\\
\hline
NC & 0.000667442 & 0.009817072 & 1.079273577\\
SP & 0.000642856 & 0.008167310 & 1.249496006\\
single-period & 0.000625088 & 0.007648792 & 1.297326189\\
2-period & 0.000546146 & 0.006535765 & 1.326516534\\
\hline
\end{tabular}
\caption{100 Stocks with transaction costs}
\label{table5}
\end{table}
The following plots are the rolling 1-year Sharpe ratios of each strategy with transaction costs, starting from different dates.
\begin{figure}[H]
\centering
\includegraphics[width=14cm]{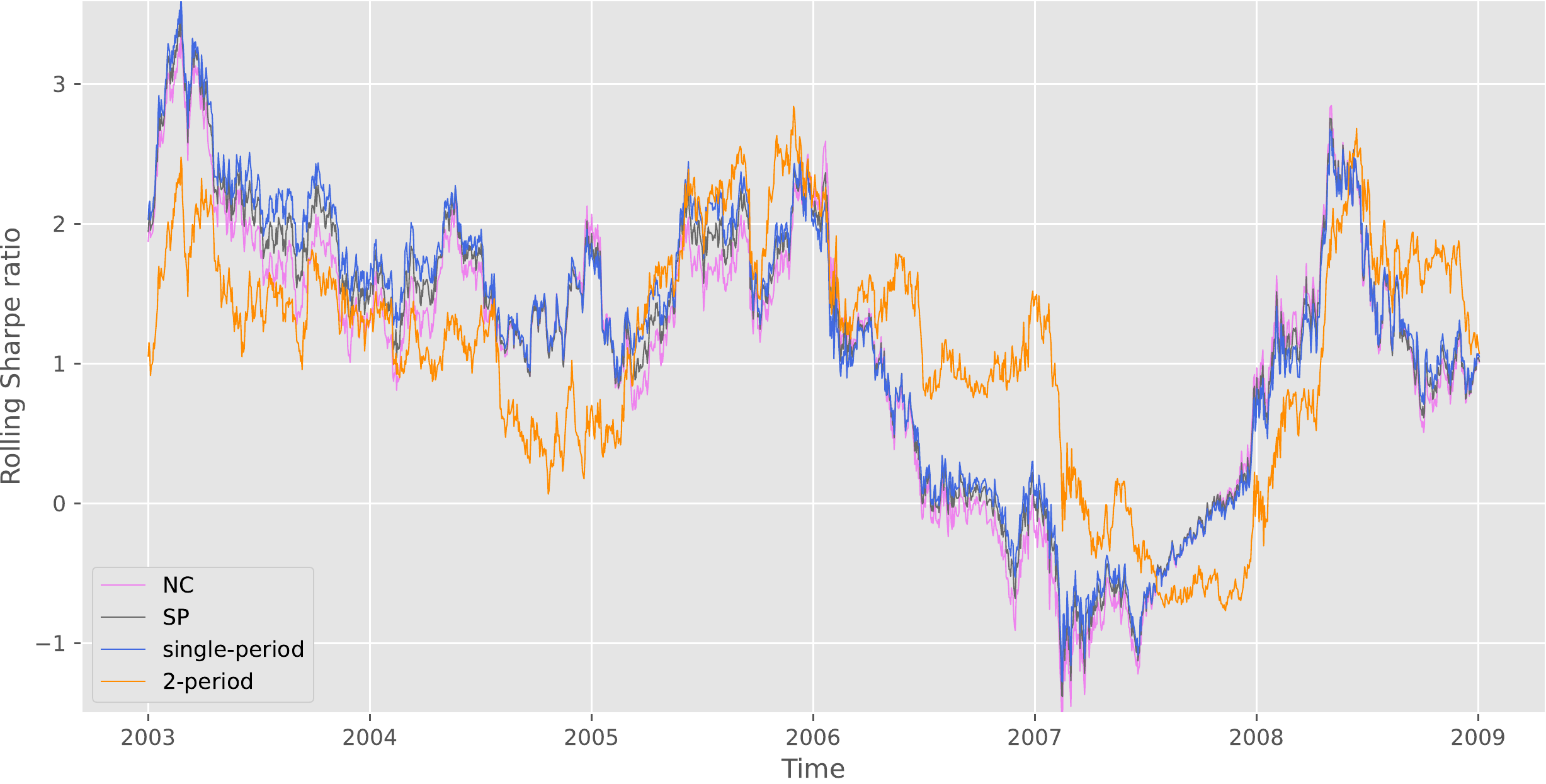}
\caption{100 Stocks with transaction costs, starting 2002.02.01}
\label{fig10}
\end{figure}
\begin{figure}[H]
\centering
\includegraphics[width=14cm]{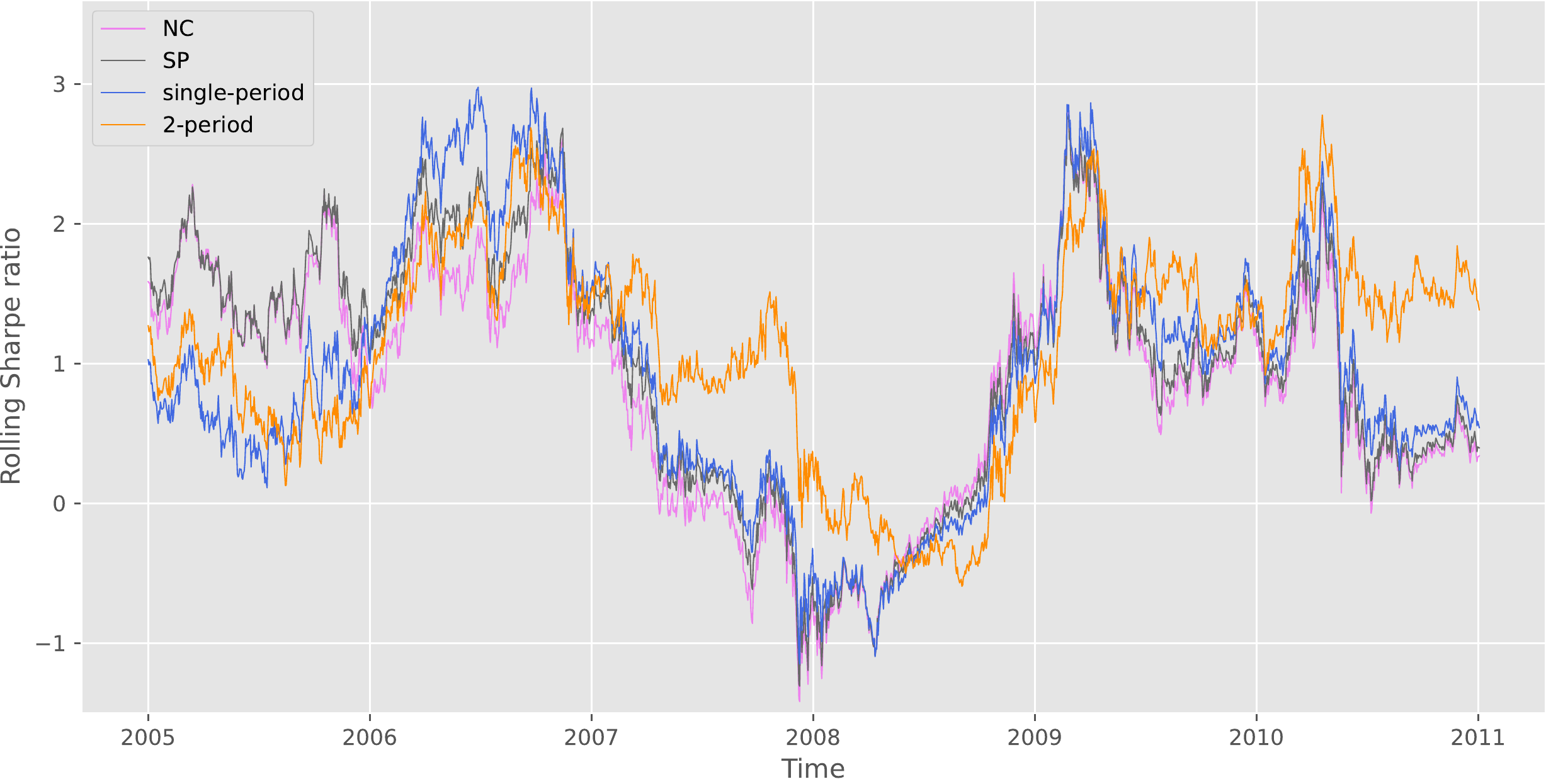}
\caption{100 Stocks with transaction costs, starting 2004.06.01}
\label{fig1.0.7}
\end{figure}
\begin{figure}[H]
\centering
\includegraphics[width=14cm]{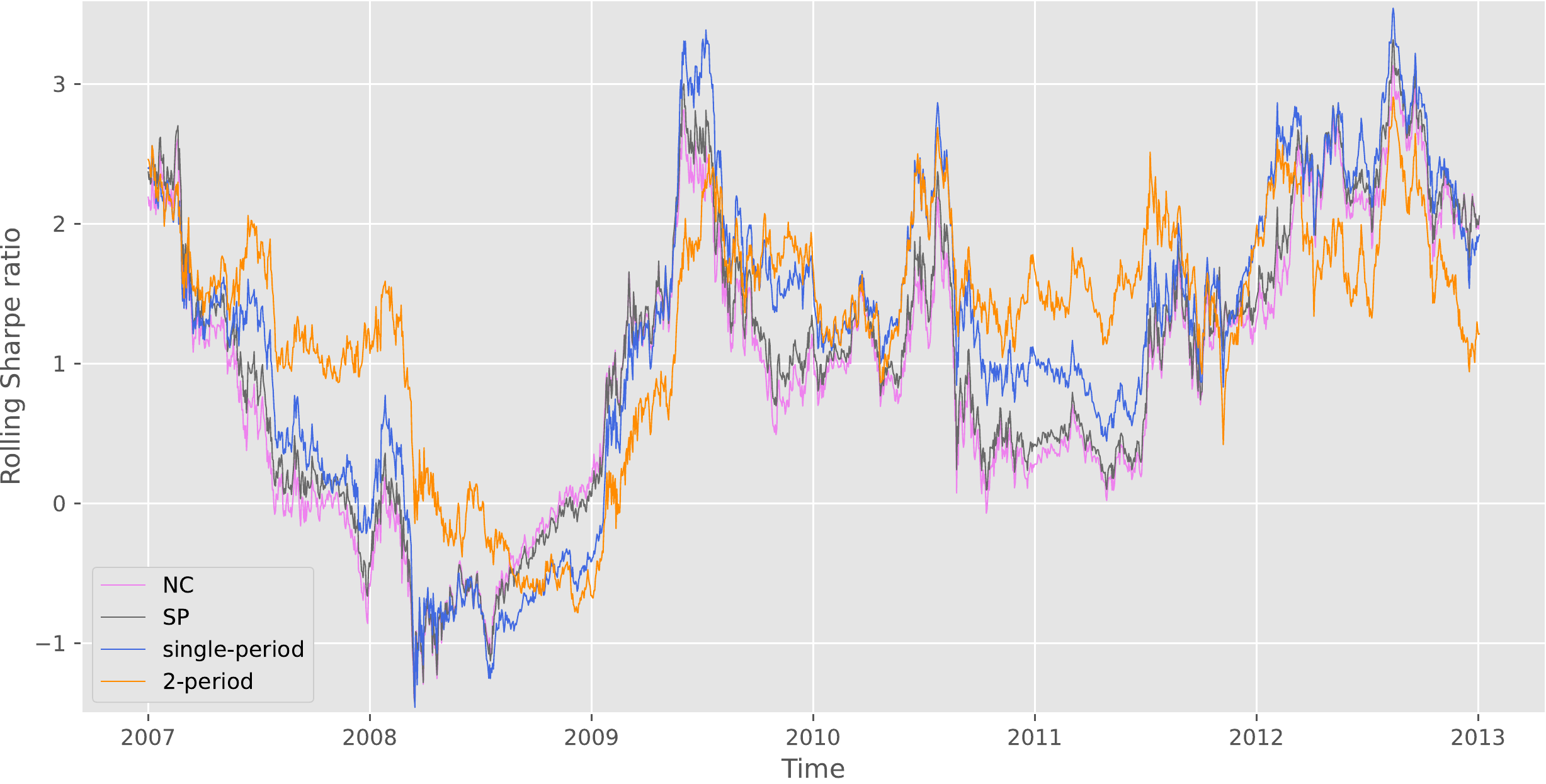}
\caption{100 Stocks with transaction costs, starting 2006.06.01}
\label{fig1.0,8}
\end{figure}

\begin{figure}[H]
\centering
\includegraphics[width=14cm]{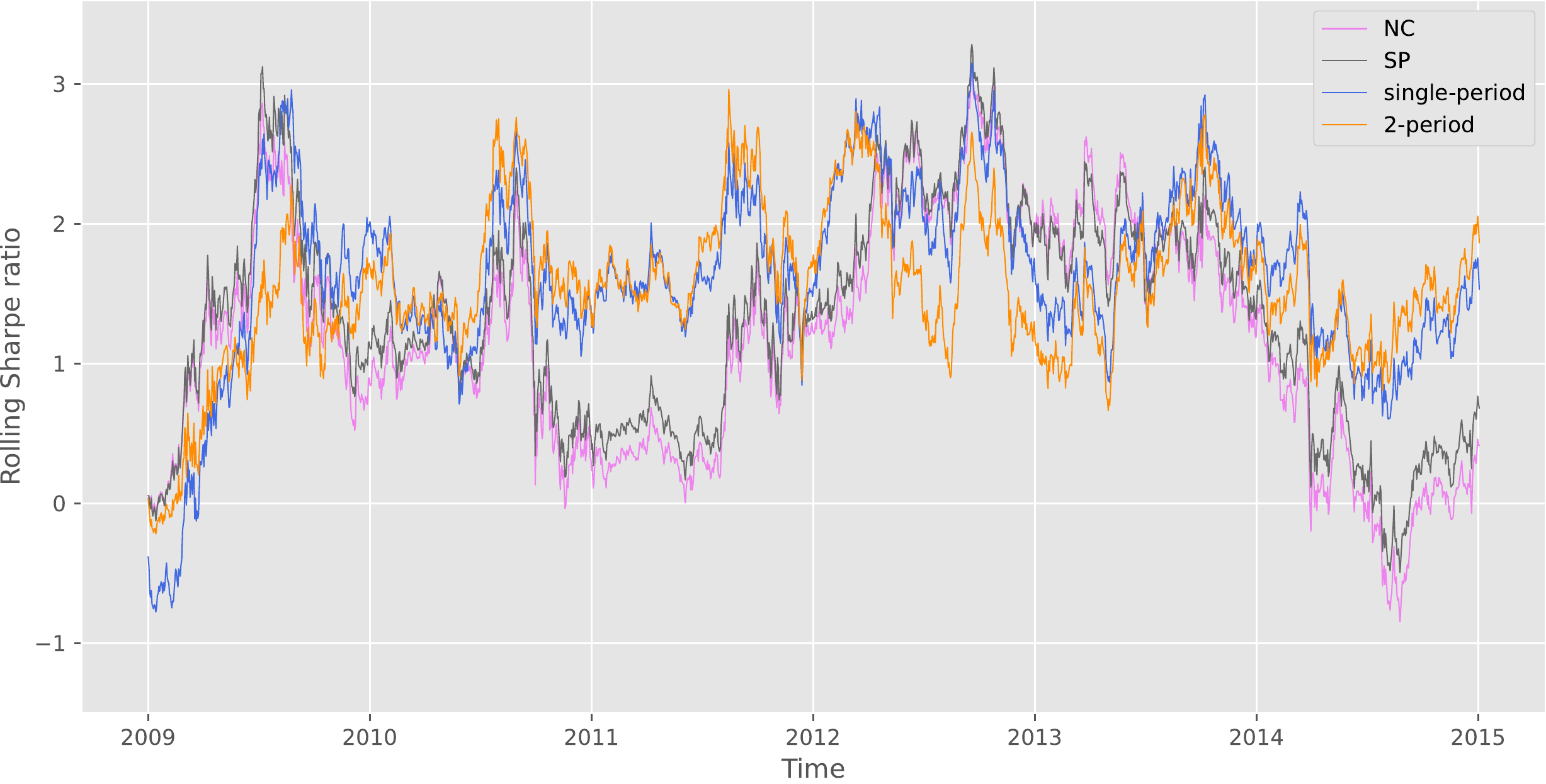}
\caption{100 Stocks with transaction costs, starting 2008.08.01}
\label{fig11}
\end{figure}
\begin{figure}[H]
\centering
\includegraphics[width=14cm]{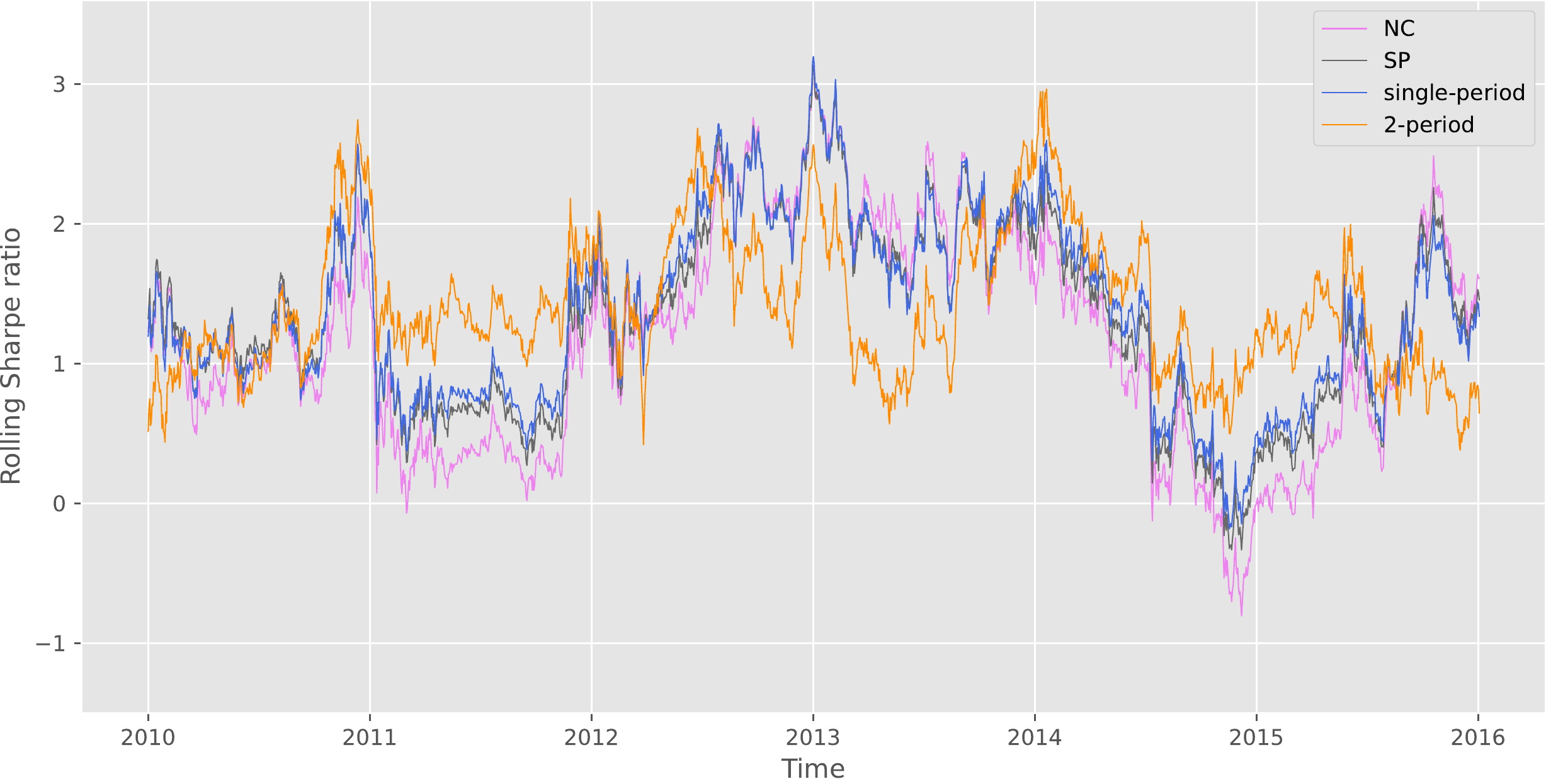}
\caption{100 Stocks with transaction costs, starting 2009.06.01}
\label{fig12}
\end{figure}

\subsection{Conclusion}
In our simulations, the 2-period model performed competitively compared to equal-weighted, classical Markowitz, UM, NC, SP, and single-period model. The simulation results are similar whether there are transaction costs or not.

For the four simulations starting from 2002.02.01, 2004.06.01, 2006.06.01, and 2009.06.01, the 2-period model showed the highest Sharpe ratio and the lowest standard deviation. The only exception is the simulation starting from 2008.08.01 that the 2-period model showed the moderately high Sharpe ratio among all strategies and consistently outperformed the equal-weighted, NC, and SP strategies. On the other hand, qualitatively speaking, the rolling 1-year Sharpe ratios of the two period model generally have smaller magnitude from a peak to a trough compared to the other strategies.

While there is a limitation that we simulate these strategies only for three time frames, our empirical results provide a promising performance of 2-period model we suggested.
\section{The Effect of the Number of Periods}
\label{section7}
We continue to think about multi-period without transaction costs and extend 2-period to the 3-period robust model under the first order Taylor expansion. Before doing the numerical simulations, we select the top 50 market cap of S\&P 500 index companies at five different time points 2002.02.01, 2004.06.01, 2006.06.01, 2008.08.01, and 2009.06.01 from the data provided by Bloomberg Terminal.\footnote{We limit our number of stocks due to the computational cost for 3-period model.} We use these 50 companies' stock daily prices to calculate the daily return. Then we input these data into our 3-period robust model and get the optimal results. Then we compare our 2- and 3-period with the single-period method and study the effect of the number of periods.

We implement the rebalancing method and the transaction costs illustrated in Section \ref{section5}.



\subsubsection*{Simulation Results}
Table \ref{table6} is the comparison between the investment strategies with different number of periods when there is no transaction costs.
\begin{table}[H]\centering
\small
\begin{tabular}{ |c||c|c|c|  }
 \hline
2002.02.01 & Mean (Daily) & Std Dev (Daily) & Sharpe (Annualized)\\
 \hline
single-period & 0.000515459 & 0.013311754 & 0.614694581\\
2-period & 0.000490590 & 0.009018402 & 0.863553332\\
3-period & 0.000556922 & 0.009458527 & 0.934697764\\
 \hline
 2004.06.01 & Mean (Daily) & Std Dev (Daily) & Sharpe (Annualized)\\
 \hline
single-period & 0.000467369 & 0.014623968 & 0.507335540\\
2-period & 0.000552433 & 0.009214212 & 0.951747573\\
3-period & 0.000537145 & 0.009232455 & 0.923580659\\
 \hline
 2006.06.01 & Mean (Daily) & Std Dev (Daily) & Sharpe (Annualized)\\
 \hline
single-period & 0.000583143 & 0.011638202 & 0.795407933\\
2-period & 0.000555751 & 0.008779726 & 1.004847740\\
3-period & 0.000587899 & 0.009221986 & 1.011996474\\
 \hline
 2008.08.01 & Mean (Daily) & Std Dev (Daily) & Sharpe (Annualized)\\
 \hline
single-period & 0.000666363 & 0.009507792 & 1.112580792\\
2-period & 0.000647658 & 0.009150105 & 1.123620969\\
3-period & 0.000699562 & 0.009448067 & 1.175394394\\
\hline
 2009.06.01 & Mean (Daily) & Std Dev (Daily) & Sharpe (Annualized)\\
 \hline
single-period & 0.000616230 & 0.007546618 & 1.296257439\\
2-period & 0.000599220 & 0.007416430 & 1.282600767\\
3-period & 0.000578595 & 0.007437629 & 1.234923724\\
\hline
\end{tabular}
\caption{50 Stocks without transaction costs}
\label{table6}
\end{table}

The following plots are the rolling 1-year Sharpe ratios of each strategy without transaction costs, starting from different dates.

\begin{figure}[H]
\centering
\includegraphics[width=14cm]{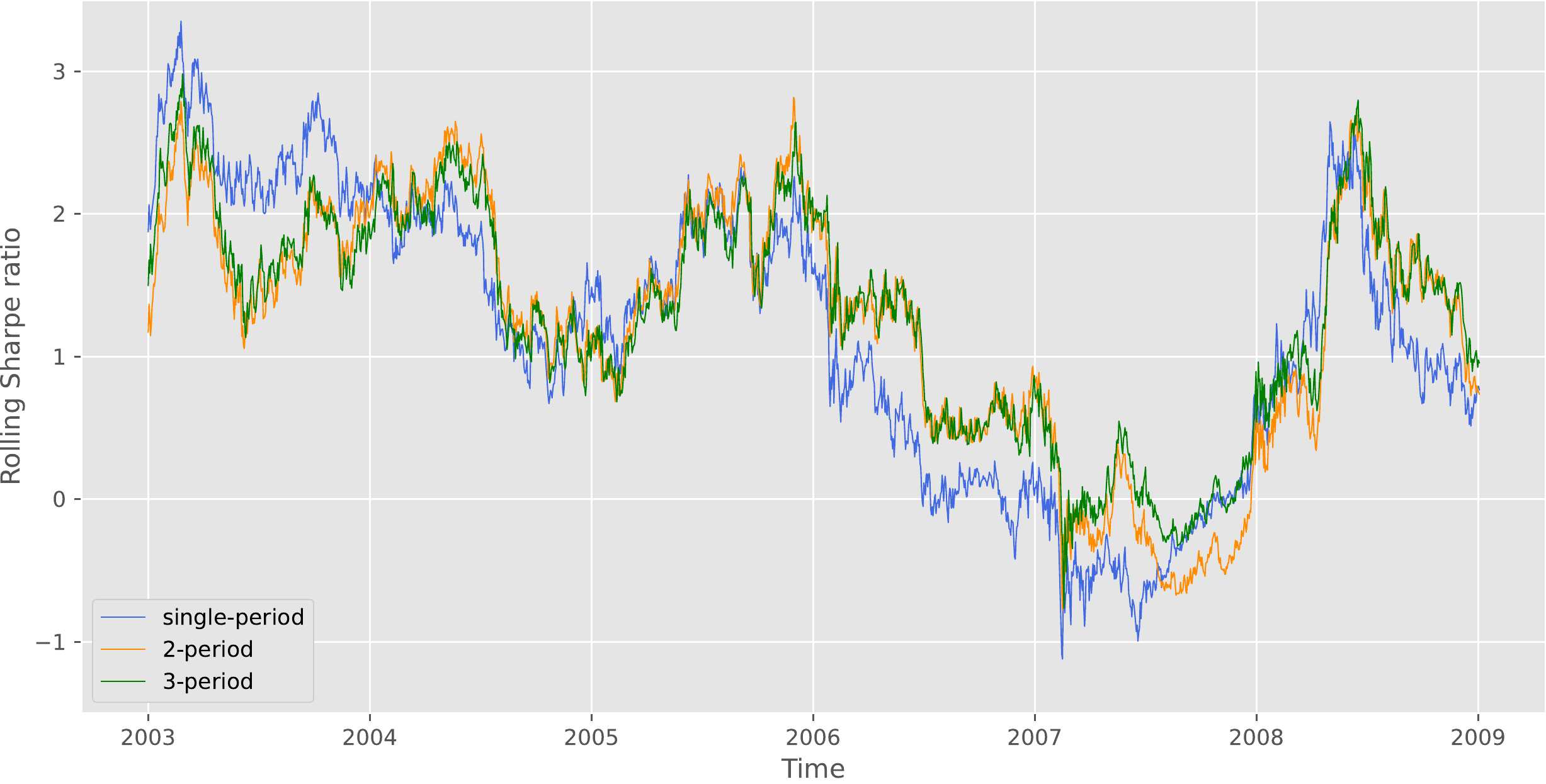}
\caption{50 Stocks without transaction costs, starting 2002.02.01}
\label{fig13}
\end{figure}
\begin{figure}[H]
\centering
\includegraphics[width=14cm]{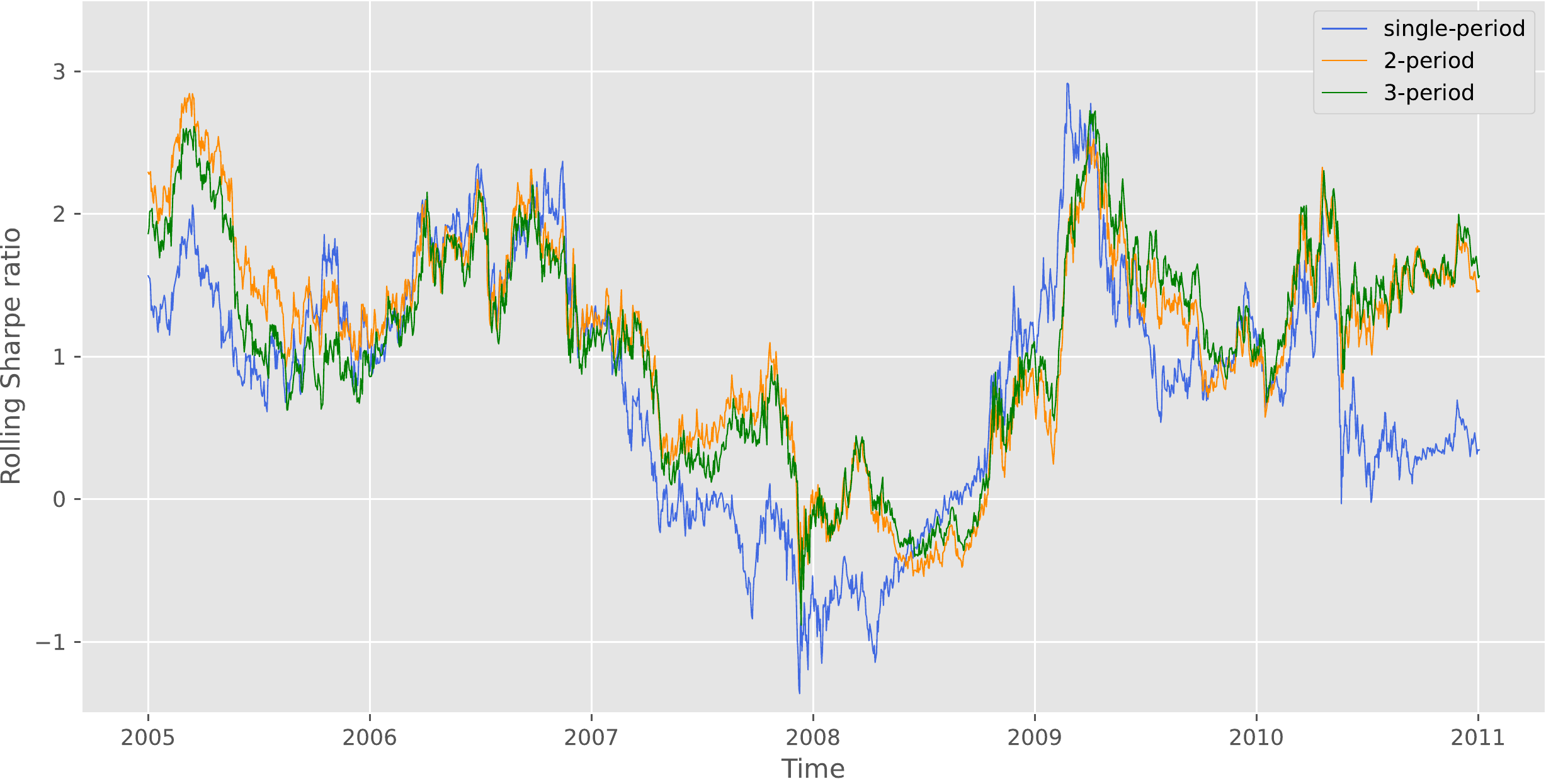}
\caption{50 Stocks without transaction costs, starting 2004.06.01}
\label{fig1.0.9}
\end{figure}
\begin{figure}[H]
\centering
\includegraphics[width=14cm]{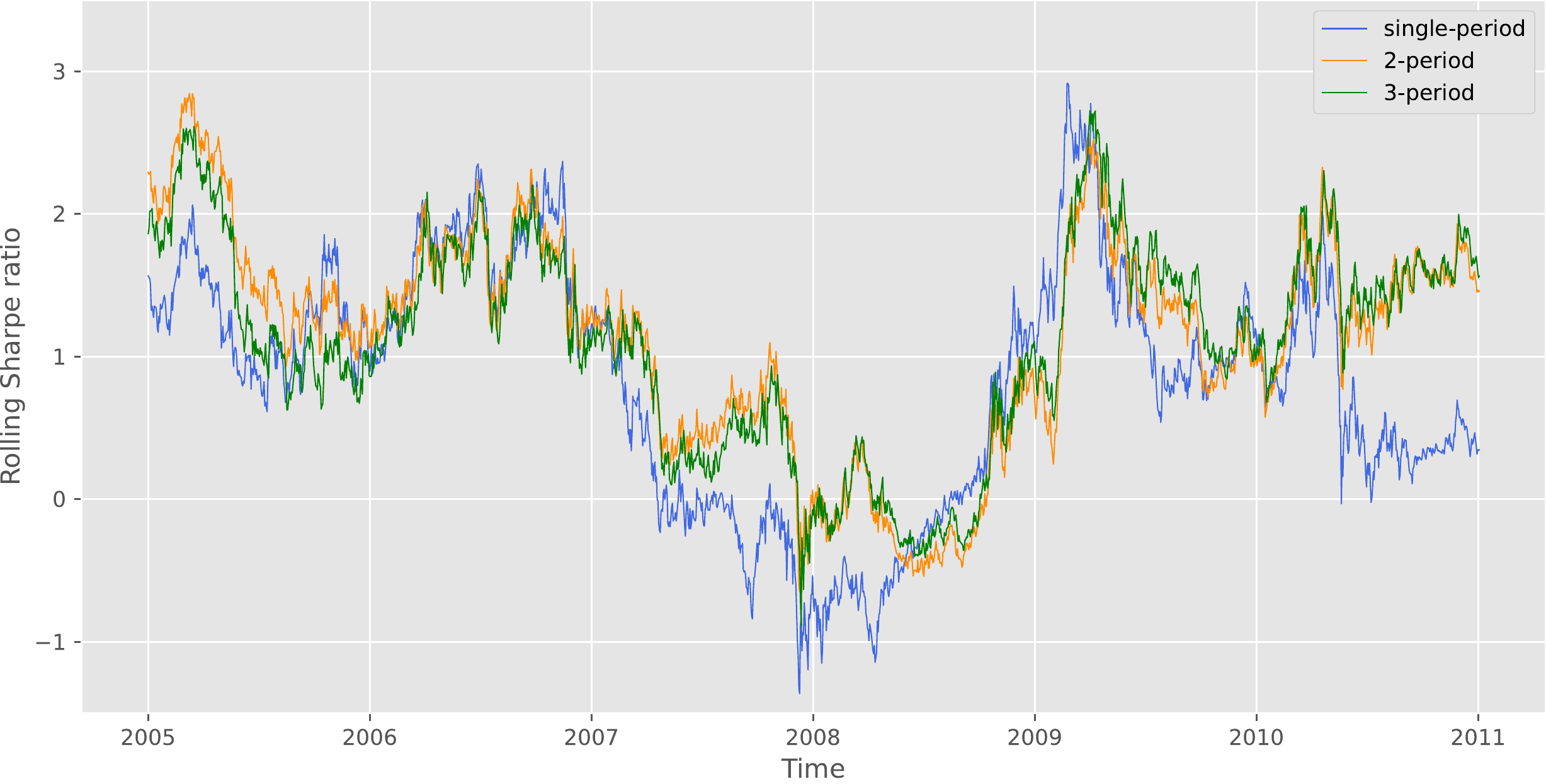}
\caption{50 Stocks without transaction costs, starting 2006.06.01}
\label{fig1.0.10}
\end{figure}

\begin{figure}[H]
\centering
\includegraphics[width=14cm]{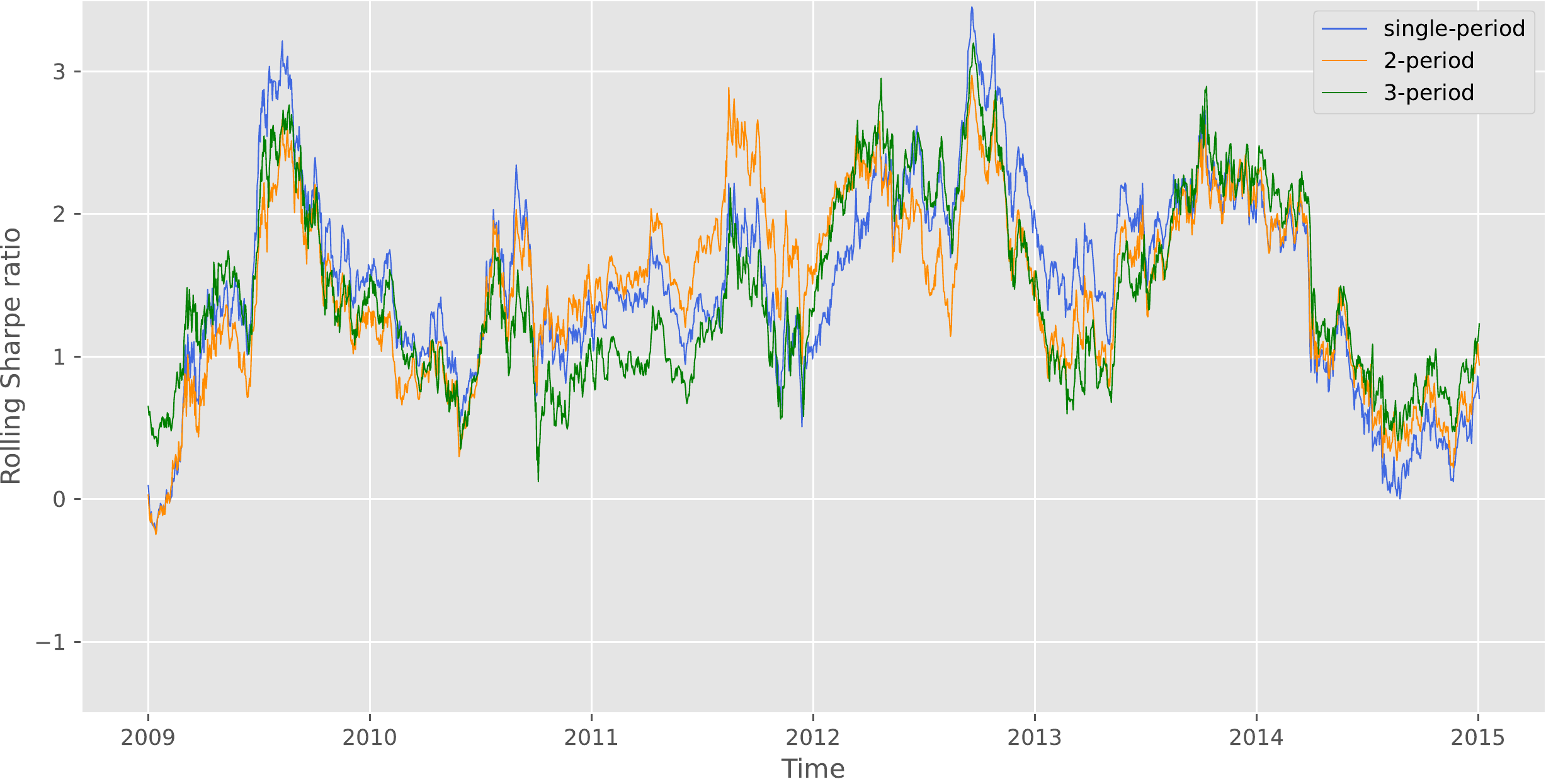}
\caption{50 Stocks without transaction costs, starting 2008.08.01}
\label{fig14}
\end{figure}
\begin{figure}[H]
\centering
\includegraphics[width=14cm]{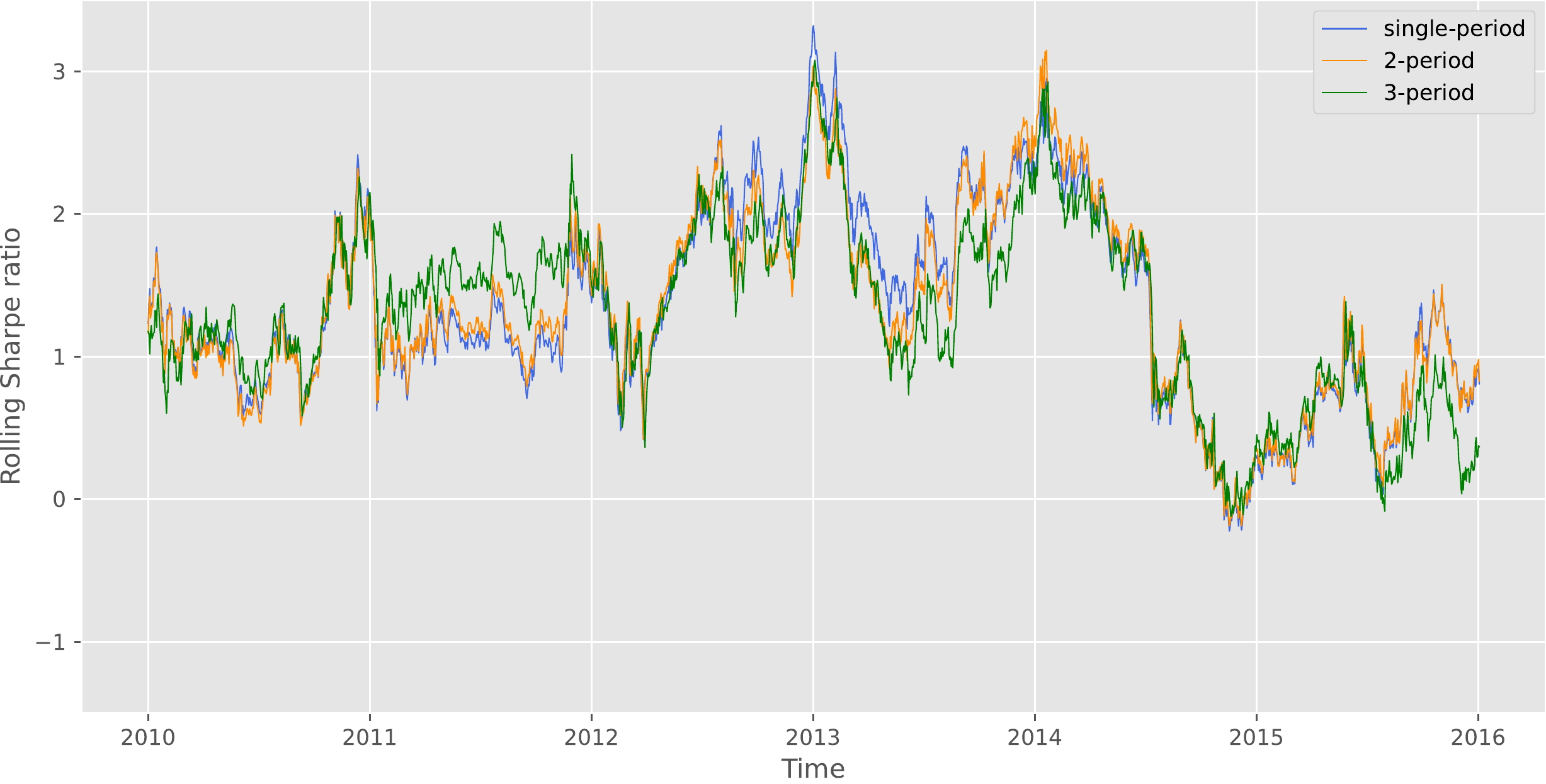}
\caption{50 Stocks without transaction costs, starting 2009.06.01}
\label{fig15}
\end{figure}
Table \ref{table7} is the comparison between the investment strategies with the transaction costs.
\begin{table}[H]\centering
\small
\begin{tabular}{ |c||c|c|c|  }
 \hline
2002.02.01 & Mean (Daily) & Std Dev (Daily) & Sharpe (Annualized)\\
 \hline
single-period & 0.000500910 & 0.013312917 & 0.597292435\\
2-period & 0.000475919 & 0.009031897 & 0.836478401\\
3-period & 0.000550335 & 0.009429044 & 0.926530660\\
 \hline
2004.06.01 & Mean (Daily) & Std Dev (Daily) & Sharpe (Annualized)\\
 \hline
single-period & 0.000463723 & 0.014620126 & 0.503510369\\
2-period & 0.000548167 & 0.009224016 & 0.943394049\\
3-period & 0.000535506 & 0.009213172 & 0.922690207\\
 \hline
2006.06.01 & Mean (Daily) & Std Dev (Daily) & Sharpe (Annualized)\\
 \hline
single-period & 0.000577600 & 0.011625840 & 0.788684216\\
2-period & 0.000545246 & 0.008792993 & 0.984366355\\
3-period & 0.000583867 & 0.009193292 & 1.008192565\\
 \hline
 2008.08.01 & Mean (Daily) & Std Dev (Daily) & Sharpe (Annualized)\\
 \hline
single-period & 0.000667794 & 0.009498350 & 1.116078218\\
2-period & 0.000647608 & 0.009162976 & 1.121955506\\
3-period & 0.000702317 & 0.009423333 & 1.183120873\\
\hline
 2009.06.01 & Mean (Daily) & Std Dev (Daily) & Sharpe (Annualized)\\
 \hline
single-period & 0.000604619 & 0.007545104 & 1.272087363\\
2-period & 0.000590169 & 0.007421212 & 1.262415012\\
3-period & 0.000572986 & 0.007426424 & 1.224797421\\
\hline
\end{tabular}
\caption{50 Stocks with transaction costs}
\label{table7}
\end{table}
The following plots are the rolling 1-year Sharpe ratios of each strategy with transaction costs, starting from different dates.
\begin{figure}[H]
\centering
\includegraphics[width=14cm]{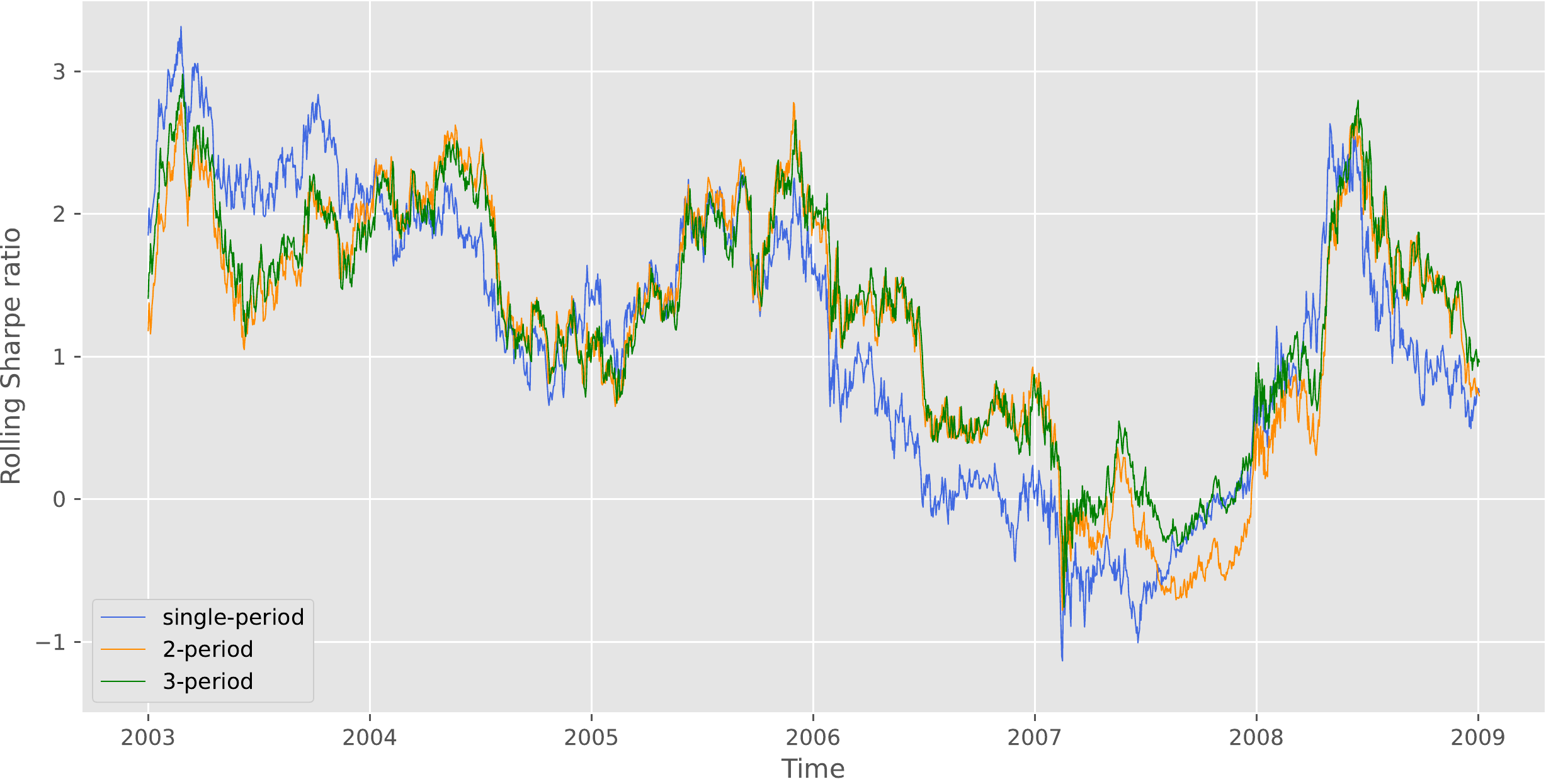}
\caption{50 Stocks with transaction costs, starting 2002.02.01}
\label{fig16}
\end{figure}
\begin{figure}[H]
\centering
\includegraphics[width=14cm]{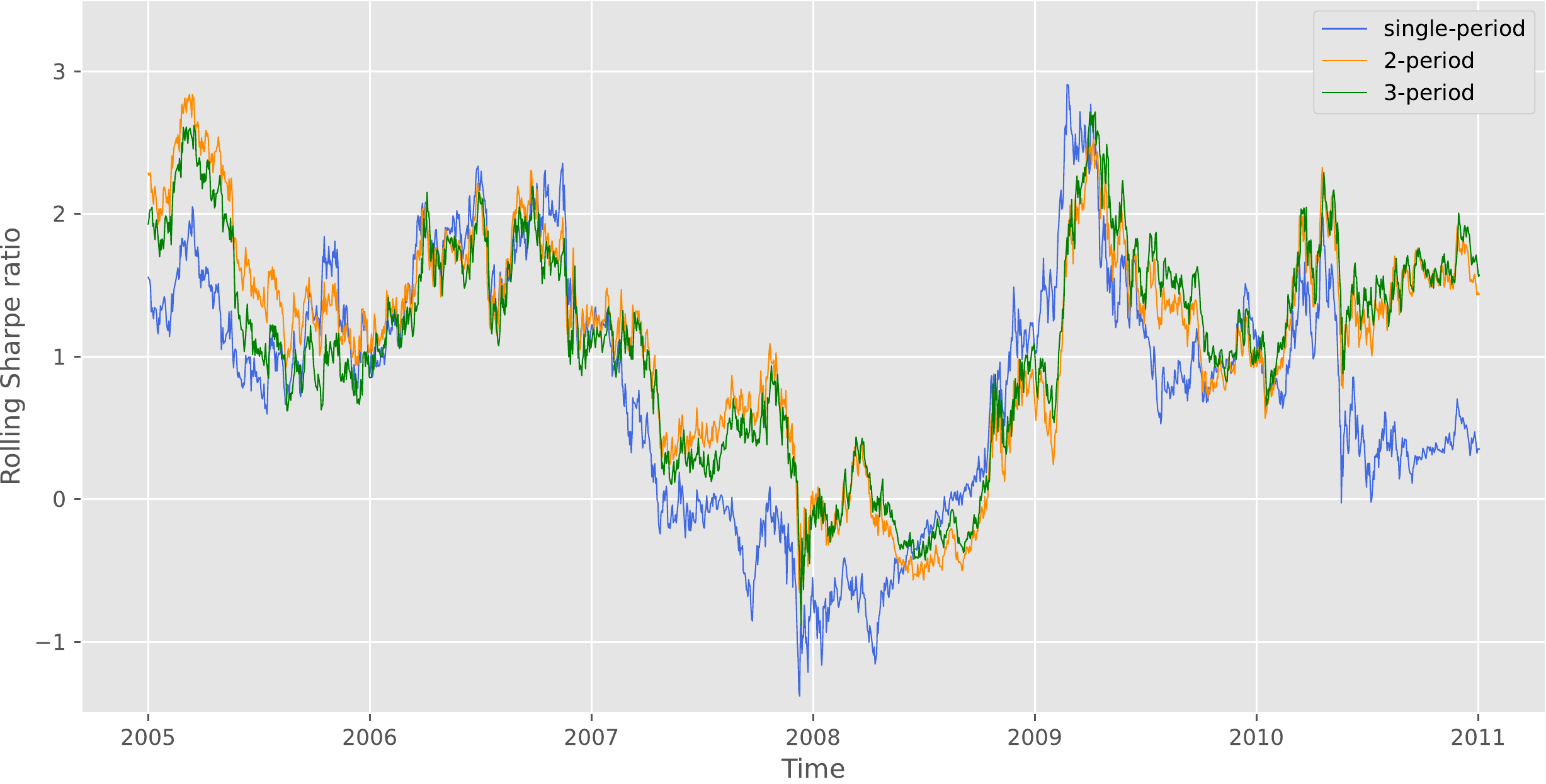}
\caption{50 Stocks with transaction costs, starting 2004.06.01}
\label{fig1.0.11}
\end{figure}
\begin{figure}[H]
\centering
\includegraphics[width=14cm]{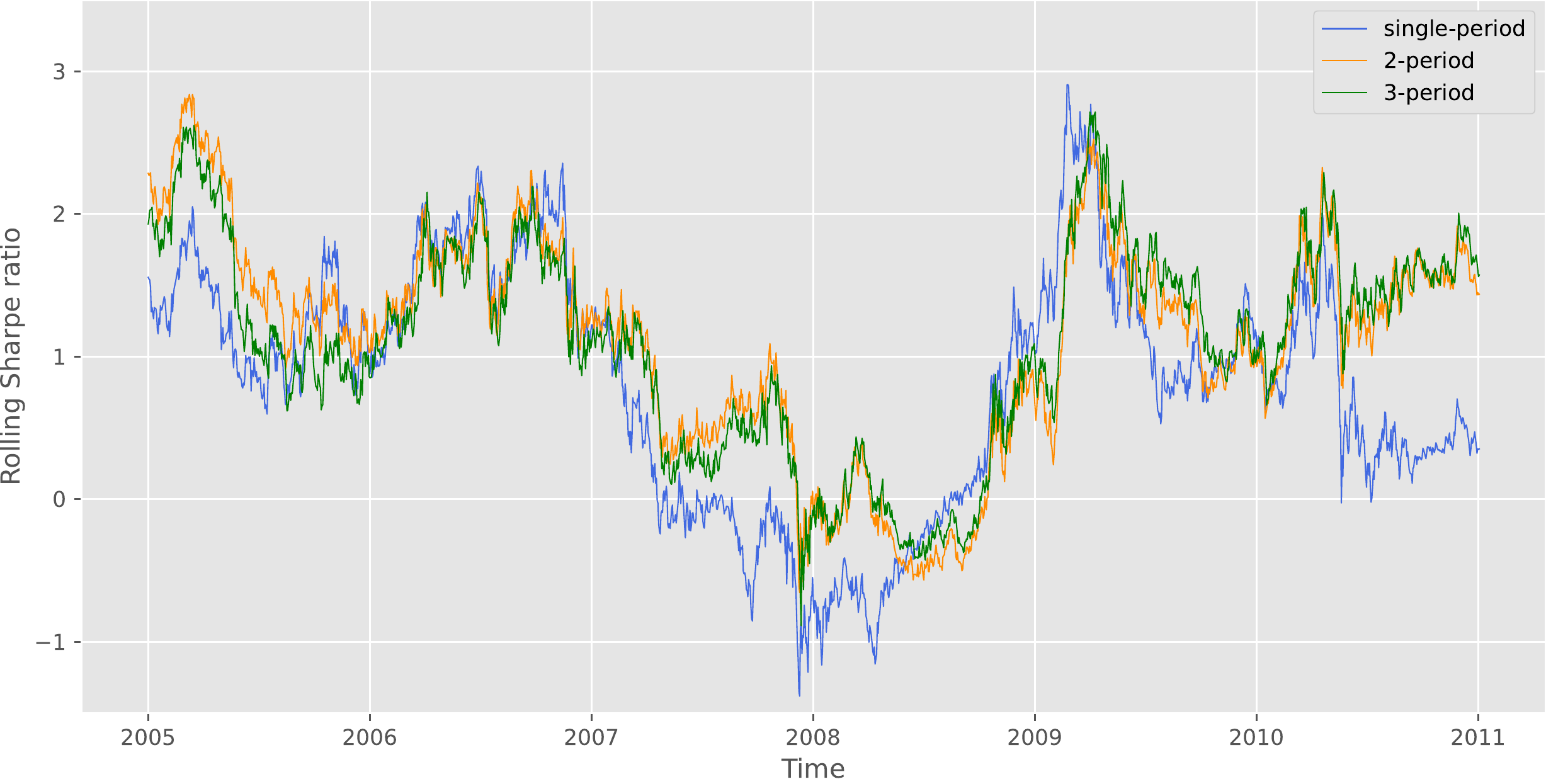}
\caption{50 Stocks with transaction costs, starting 2006.06.01}
\label{fig1.0.12}
\end{figure}

\begin{figure}[H]
\centering
\includegraphics[width=14cm]{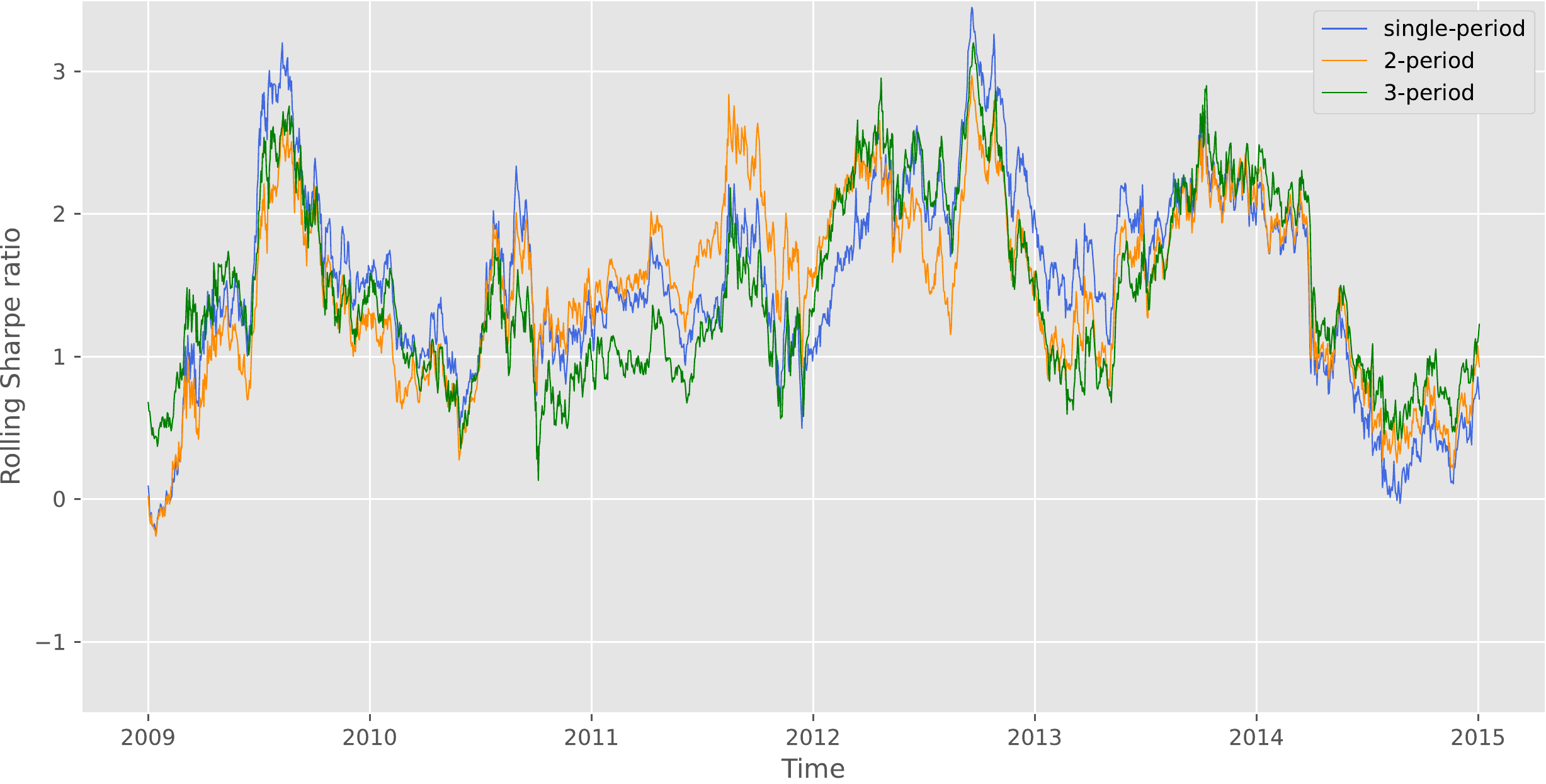}
\caption{50 Stocks with transaction costs, starting 2008.08.01}
\label{fig17}
\end{figure}
\begin{figure}[H]
\centering
\includegraphics[width=14cm]{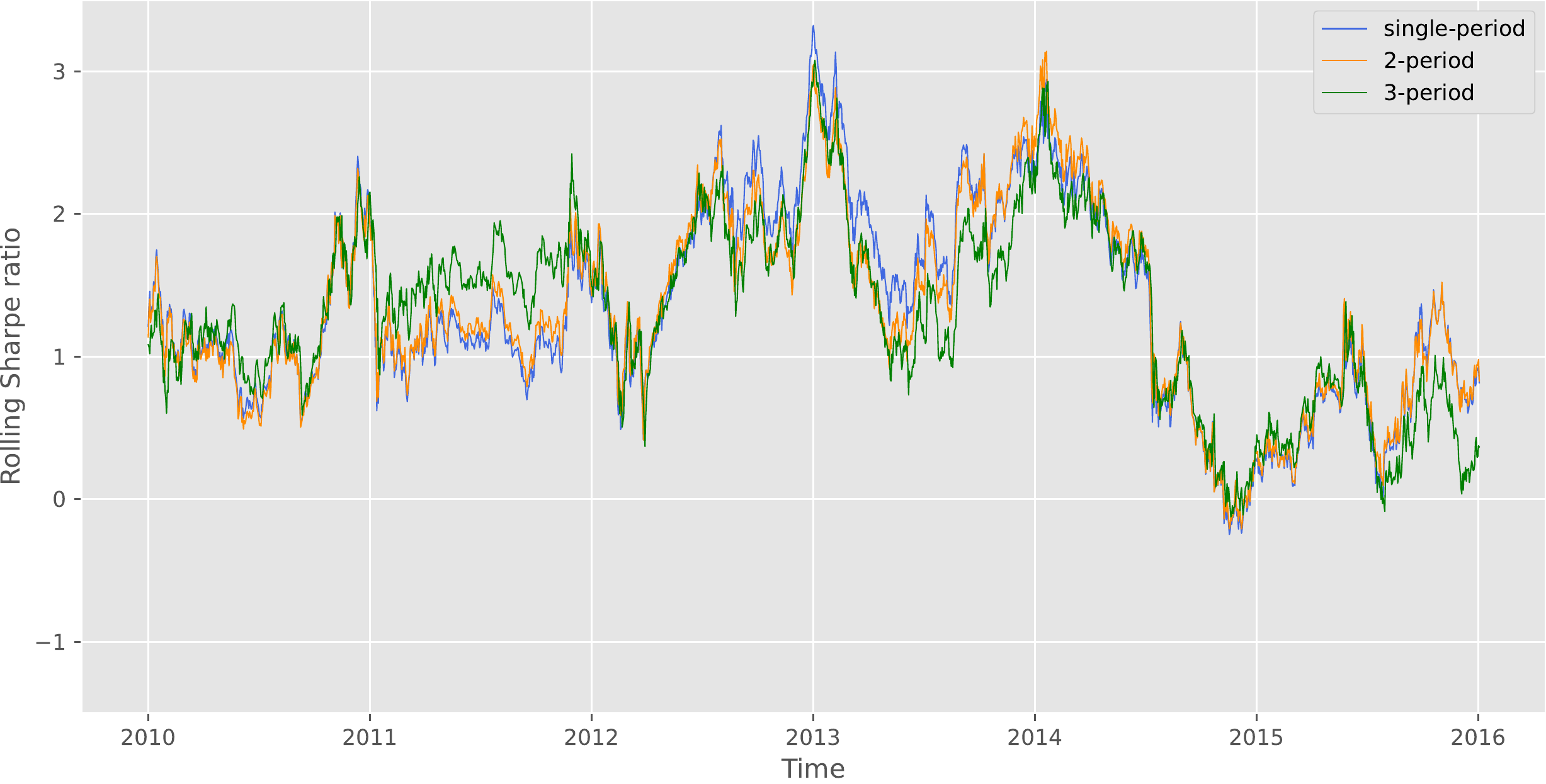}
\caption{50 Stocks with transaction costs, starting 2009.06.01}
\label{fig18}
\end{figure}

The following figure is the change in Sharpe ratio as the number of periods increases.
\begin{figure}[H]
\centering     
\subfigure[Without transaction costs]{\label{fig25a}\includegraphics[width=7.8cm]{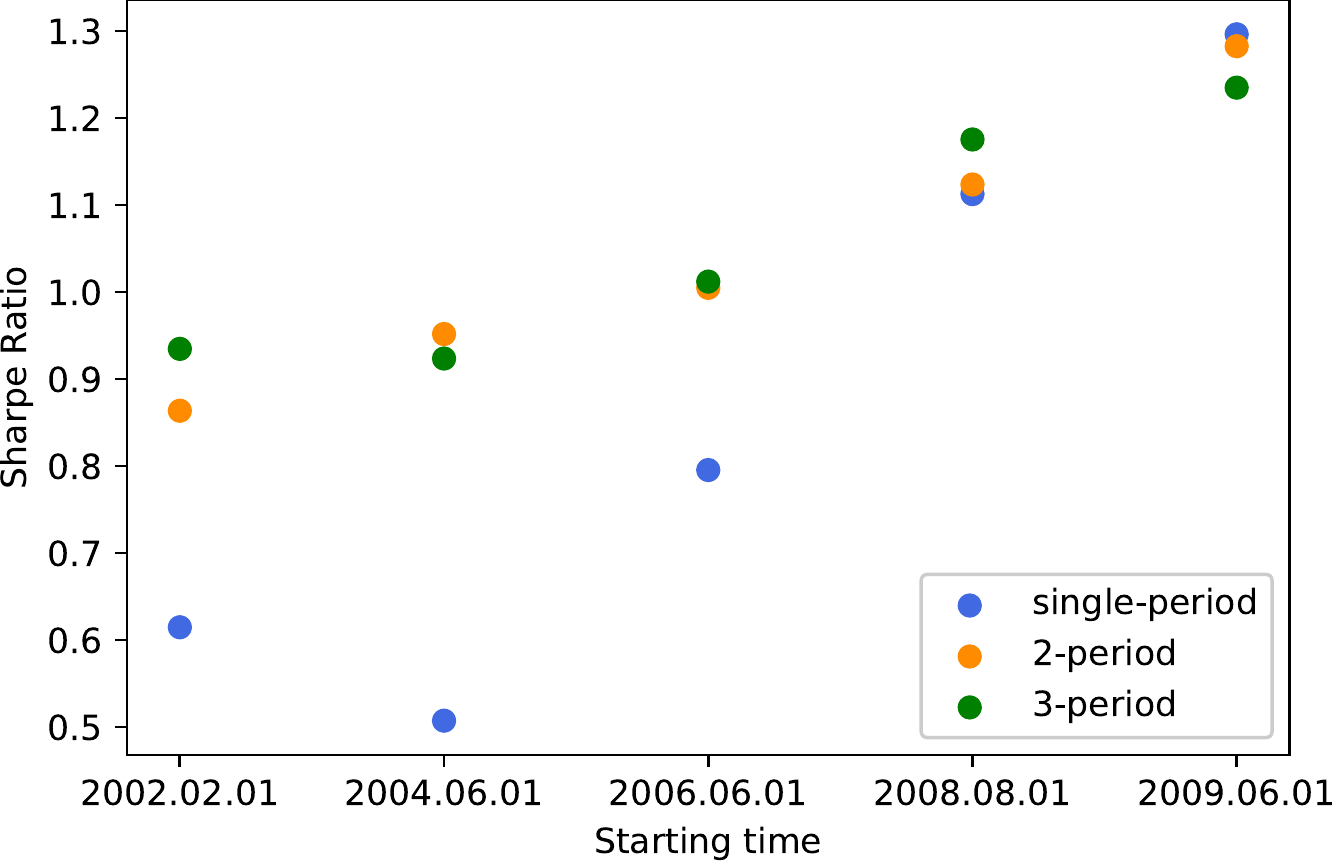}}
\subfigure[With transaction costs]{\label{fig25b}\includegraphics[width=7.5cm]{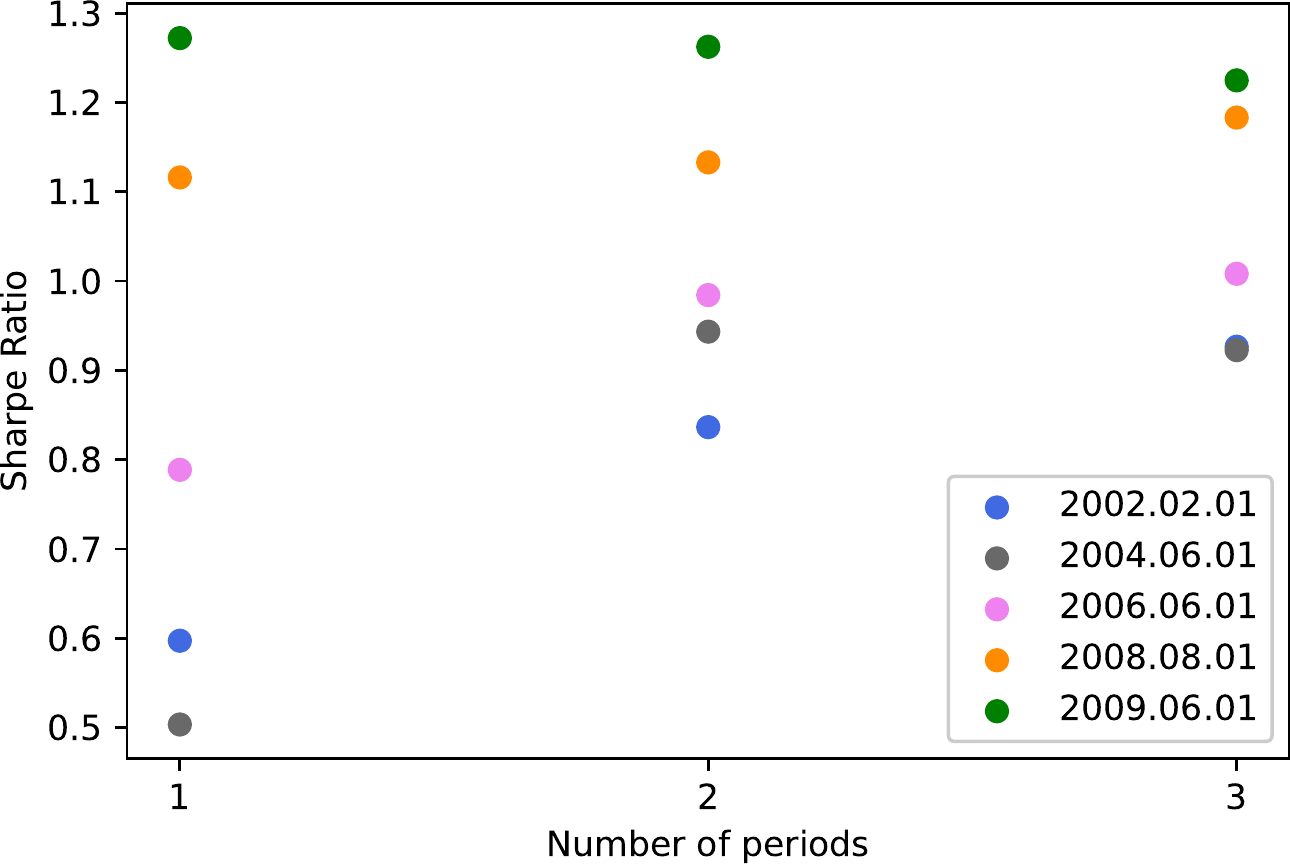}}
\label{fig25}
\caption{Sharpe ratio VS Number of periods}
\end{figure}

\subsection*{Conclusion}
According to our empirical studies, as the number of periods increases, the Sharpe ratio is likely to increase as well: see Figure \ref{fig25a} and Figure \ref{fig25b}. The two exceptions are the simulations starting from 2004.06.01 and 2009.06.01. For the simulation starting from 2004.06.01, the 2-period model showed the best Sharpe ratio and single-period model showed the worst one. In the other case starting from 2009.06.01, the single-period model performed the best and 3-period model performed the worst.

As the number of periods in the model increases, the parameters we need to estimate and the computational cost increases dramatically. Depending on the in-sample data we used for parameter estimation, it is possible that 3-period performs worse than single- or 2-period model.

While we need more experiments, our simulations suggest 2- or 3-period is desirable in terms of the portfolio performance to apply our multiperiod model to the real-world data.
\section{Conclusion}
\label{section8}
In this article, we study robust mean-variance optimization in multiperiod portfolio selection. We allow the true probability measure to be inside a Wasserstein ball that is specified by the empirical data and a given confidence level. We represent the optimal control as a function of historical stock prices and approximate it using Taylor expansion, which allows us to extend the single-period to a multiperiod model. Then, we apply our framework to some numerical simulations of the US stock market, which provides a promising result compared to other popular strategies.

\newpage
\bibliographystyle{abbrvnat}
\bibliography{ethan}

\newpage
\begin{appendices}
\section{}
\label{appendix}
We need to prove some components of $A$ are $0$ does not change the value of the optimization problem (\ref{eq15}).

Our initial optimization problem is
\begin{equation*}
	\begin{aligned}
	\inf_{\boldsymbol{\pi}\in\scF_{(\delta,\Bar{\alpha})}}\max_{\lambda\geq \bar\alpha}\brak{\sup_{P\in \scU_\delta(Q),E_P[\boldsymbol{\pi}^{\intercal}\boldsymbol{R}]=\lambda}E_P\edg{\brak{\boldsymbol{\pi}^\intercal\boldsymbol{R}}^2-\lambda^2}}.
	\end{aligned}
\end{equation*}
After using the Taylor expansion to the investment strategy, for given $\boldsymbol{\pi}$, the middle and inner problems finally become
\eqn{\Big(\sqrt{A^\intercal Var_Q(M)A}+\sqrt{\delta}\|A\|_2\Big)^2.}
We know $A$ has the restriction that its some components are $0$. We can get $\boldsymbol{\pi}$ from $A$ and historical data $N$ (see Section \ref{section4}). Then $\boldsymbol{\pi}\approx A^\intercal N$.

So we can rewrite the outer problem as 
\eqn{\inf_{ A\in\scF_{(\delta,\Bar{\alpha})}}\Big(\sqrt{A^\intercal Var_Q(M)A}+\sqrt{\delta}\|A\|_2\Big)^2.}
Denote $A_0$ is all $0$ components in $A$ and the rest components of $A$ is $A_1$. We rearrange the structure of $A$ and $M$ and keep the value of $A^\intercal M$ unchanged. 

Let $A=(A_0,A_1)$ and $M=(M_0,M_1)$, where $M_0$ is some components of $M$ corresponding to $A_0$ in $A$, $M_1$ is the rest components of $M$ corresponding to $A_1$ in $A$. $A_0$ and $M_0$, $A_1$ and $M_1$ have the same dimension respectively.

Therefore, the above middle and inner problems can be written as
\eqn{\Big(\sqrt{(A_0,A_1)^\intercal Var_Q((M_0,M_1))(A_0,A_1)}+\sqrt{\delta}\|(A_0,A_1)\|_2\Big)^2,}
which is equal to 
\eqn{\Big(\sqrt{A_1^\intercal Var_Q(M_1)A_1}+\sqrt{\delta}\|A_1\|_2\Big)^2.}
For the feasible region (see Section \ref{section4})
\eqn{\mathcal{F}_{(\delta,\Bar{\alpha})}:=\Big\{A:E_{Q}[A^{\intercal}M]-\sqrt{\delta}\|A\|_2\geq\Bar{\alpha}\Big\},}
it becomes
\eqn{\mathcal{F}_{(\delta,\Bar{\alpha})}:=\Big\{A_1:E_{Q}[A_1^{\intercal}M_1]-\sqrt{\delta}\|A_1\|_2\geq\Bar{\alpha}\Big\}.}
So the final optimization problem is 
\eqn{\inf_{ A_1\in\scF_{(\delta,\Bar{\alpha})}}\Big(\sqrt{A_1^\intercal Var_Q(M_1)A_1}+\sqrt{\delta}\|A_1\|_2\Big)^2}
or
\eqn{\inf_{ A_1\in\scF_{(\delta,\Bar{\alpha})}}\sqrt{A_1^\intercal Var_Q(M_1)A_1}+\sqrt{\delta}\|A_1\|_2.}
\end{appendices}
\end{document}